%% file: main.tex
\documentclass[12pt]{article}
\usepackage{authblk}
\usepackage{graphicx}
\usepackage{tikz}
\usetikzlibrary{patterns} 
\usepackage{dsfont} 
\usepackage{amsfonts}
\usepackage{mathrsfs}
\usepackage{float}
\usepackage{mathtools}
\usepackage{amssymb}
\usepackage{enumitem}
\usepackage{amsthm}
\usepackage[thinc]{esdiff}
\usepackage{tkz-euclide}
\usepackage{tikz-cd}
\usepackage{hyperref}
\hypersetup{
    colorlinks=true,
    linkcolor=black,
    citecolor=black,
    filecolor=black,      
    urlcolor=black}
\usepackage[capitalize,nameinlink]{cleveref}
\crefname{section}{Sect.}{Sect.}
\crefname{proposition}{Prop.}{Prop.}    
\crefname{lemma}{Lm.}{Lm.}  
\crefname{theorem}{Thm.}{Thm.}
\usepackage[margin=1.99cm]{geometry}

\theoremstyle{definition}
\newtheorem{proposition}{Proposition}[section]
\theoremstyle{definition}
\newtheorem{lemma}[proposition]{Lemma}
\theoremstyle{definition}
\newtheorem{theorem}[proposition]{Theorem}
\theoremstyle{definition}
\newtheorem*{remark}{Remark}

\newcommand{\cc}{\mathbb{C}}
\newcommand{\rr}{\mathbb{R}}
\newcommand{\tr}{\operatorname{tr}}
\newcommand{\ch}{\operatorname{ch}}
\newcommand{\rk}{\operatorname{rk}}
\def\e{\mathrm{e}}
\def\i{\mathrm{i}}

\usepackage[title]{appendix}
\numberwithin{equation}{section}

\makeatletter
\renewenvironment{proof}[1][\proofname]{%
  \par\pushQED{\qed}\normalfont%
  \topsep6\p@\@plus6\p@\relax
  \trivlist\item[\hskip\labelsep\bfseries#1\@addpunct{.}]%
  \ignorespaces
}{%
  \popQED\endtrivlist\@endpefalse
}
\makeatother

\begin{document}

\bibliographystyle{plain}

\title{Galilei covariance of the theory of Thouless pumps}

\author[1]{Tilman Esslinger}
\author[2]{Gian Michele Graf}
\renewcommand*{\Authands}{, }
\author[2,*]{Filippo Santi}
\renewcommand{\Affilfont}{\mdseries \itshape}
\affil[1]{Institute for Quantum Electronics, ETH Zurich, 8093 Zurich, Switzerland}
\affil[2]{Institute for Theoretical Physics, ETH Zurich, 8093 Zurich, Switzerland}
\affil[*]{Corresponding author: fsanti@ethz.ch}
\date{\today}
\maketitle

\begin{abstract}

\noindent The Thouless theory of quantum pumps establishes the conditions for quantized particle transport per cycle, and determines its value. When describing the pump from a moving reference frame, transported and existing charges transform, though not independently. This transformation is inherent to Galilean space and time, but it is underpinned by a transformation of vector bundles. Different formalisms can be used to describe this transformation, including one based on Bloch theory. Depending on the chosen formalism, the two types of charges will be realized as indices of either the same or different kinds. Finally, we apply the bulk-edge correspondence principle, so as to implement the transformation law within Büttiker's scattering theory of quantum pumps.
\end{abstract}


\section{Introduction}
A quantum pump can be modeled in terms of a 1-dimensional Fermi gas moving in a periodic potential that is periodic in time and, quite conveniently but not necessarily, in space as well. Thouless showed \cite{thouless} that charge transport is quantized, provided the Fermi level remains in a gap throughout the pumping cycle. The system is supposed to depend parametrically on physical time $t$, albeit slowly, so that it can be described in the adiabatic approximation. Very controlled Thouless pumping has recently been realized in experiments with ultracold atoms using dynamically controlled superlattice potentials from interfering laser beams, where the drift of the atomic cloud as a result of the pumping was observed and the topological invariance could be extracted \cite{Nakajima,Lohse,Citro}. Beyond this, the robustness against disorder \cite{Nakajima2021}, the effects of atom-atom interactions \cite{Walter}, as well as a drift reversal in the presence of an external trapping potential \cite{zhu} have been observed.\\

Realizations of Thouless pumps that are based on solid state devices are naturally described in their rest frame. By contrast, potentials that are shaped by laser beams through their changing interference patterns do not single out any particular frame of reference. To illustrate the idea, consider two superimposed optical lattices,
creating the potential
\begin{equation*}
  V(x,t)= \sum_{i=1,2}V_i\sin^2(k_ix-\omega_it)\,,
  \end{equation*}
where $k_1$, $k_2$ are commensurate. In a reference frame having
relative velocity $v$, the two
frequencies are $\hat{\omega}_i=  \omega_i-k_iv$. By suitable choice
of $v$, one can make either frequency vanish, but not both, as a rule. Such conditions call for the covariance of the theory of Thouless pumps w.r.t. Galilei transformations. Quite intuitively though, time periodicity could be lost in a generic moving frame and is preserved only for specific velocities between frames.
\begin{figure}[H]
\centering
\includegraphics[scale=0.5]{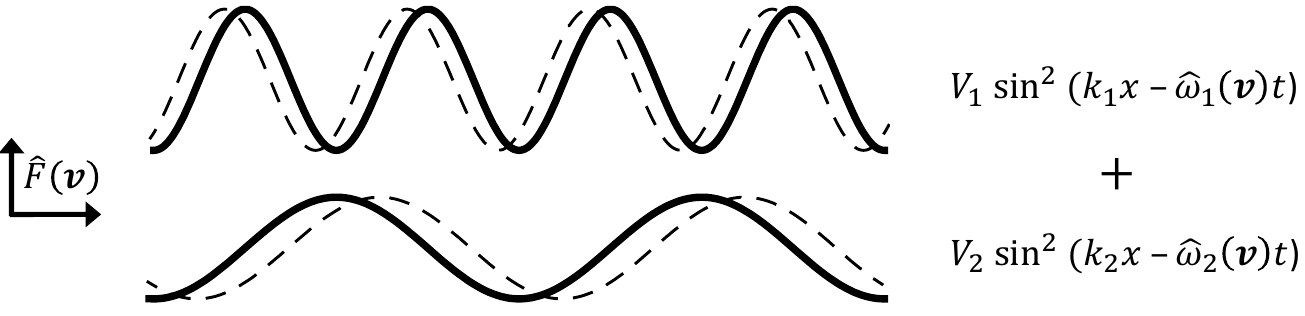}
\caption{ Illustration of periodic potentials $V_1$ and $V_2$ forming a Thouless pump, viewed from a reference frame $\hat{F}$ moving with constant velocity $v$. Dashed lines indicate potentials at a later time.}
\end{figure}
Let us consider the family of Hamiltonians on the real line 
\begin{equation}\label{hamiltonian}
H(s)=p^2 + V(x,s)\,, 
\end{equation}
smoothly parametrized by the adiabatic time $s=\varepsilon t$, ($\varepsilon\ll 1$), and where the potential $V(x,s)=V(x,s)^{*}$ has the stated periodicities:
\begin{align}
V(x,s+T)&=V(x,s)\,, \label{periodicity1}\\  
V(x+L,s)&=V(x,s)\,. \label{periodicity2} 
\end{align}
There is room in (\ref{hamiltonian}) for several channels, hence the wavefunctions may have several components and $V$ be matrix-valued.\\

The charge $Q$ transported during a period $T$ across a fixed fiducial point can be expressed in topological terms. Indeed \cite{thouless}, it is equal to the Chern number,
\begin{equation}\label{chern}
Q=\ch P\,, \qquad \ch P \coloneqq \frac{\mathrm{i}}{2\pi} \oint ds\oint dk \tr(P[\partial_s P,\partial_k P])\,,
\end{equation}
where $s$ and $k$ run over a time period and the Brillouin zone (a circle), respectively; and $P$ is the projection onto the states below the Fermi level, inducing projections $P(s,k)$. More conveniently, $P$ should be viewed as a vector bundle over a 2-torus (the Fermi bundle) arising as a subbundle $P\subset E$ of the Bloch bundle $E$, as will be explained in \cref{section 2}. For now, we just add that \cref{chern} refers to the Bloch trivialization of $E$. Another expression for $\ch P$, not relying on any trivialization, is
\begin{equation}\label{12}
\ch P=\frac{1}{2 \pi} \oint ds \oint dk \tr(P [[P, x], \partial_s P])\,,
\end{equation}
where all operators under the trace, except for $x$, but including $[P,x]$, act fiberwise on the fibers of $P$.\\

Less surprisingly, another physical quantity that is expressed in topological terms, besides $Q$, is the charge residing in a spatial period $L$, i.e. the number of occupied bands. It equals the rank of $P$,
\begin{equation}
N=\rk P\,, \qquad \rk P=\dim P_{s,k}\,,\qquad (\text{any } s,k)\,.    
\end{equation}
Let us point out similarities and differences between the \textit{transported} and the \textit{existing charges}, $Q$ and $N$. For one thing, $Q$ is based on a temporal period and on a fixed fiducial point in space. Likewise, $N$ is based on a spatial period, viewed at a fixed time. Any displacement of the temporal or spatial period or of the reference points in space or time does not affect either $Q$ or $N$. Still, the two corresponding indices, $\ch P$ and $\rk P$, are of a different kind, in that only the latter can be determined from a single fiber of the bundle. They are referred to as \textit{strong} and \textit{weak} indices of the problem at hand, following the general terminology, cf. e.g. \cite{hasan}. More generally, restricting the base space, e.g. the 2-torus, to one of smaller dimension $d$ induces a restriction of the vector bundle, thereby giving rise to weak indices. Here, $\rk P$ arises in $d=0$.\\ 

The main result, to be presented in the next section, is that the indices mix under Galilei transformations in a way that establishes the covariance of the theory.\\

The structure of the article is as follows. In \cref{section 2} we will address Galilean relativity for various purposes, such as the transformation of transported and existing charges as it is brought about by classical considerations, but also and foremost about how those transformations are embodied in the Thouless theory of pumps; in particular, how the above bundles transform. We will moreover give two independent proofs of the main result, based on \cref{chern,12}, respectively. In \cref{section 3} we shall recall an approach to quantum pumps that forgoes periodicity in space, while retaining it in time, thereby encoding the transported charge $Q$ by means of an index not relying on Bloch theory. If however space periodicity is added to the scheme, then both $Q$ and $N$ will be realized by indices of the same kind, and indeed as Chern numbers of two 2-dimensional subtori of the same 3-dimensional torus. In \cref{section 4}, we will consider a truncated system and relate some invariants pertaining to the scattering problem at the edge of the pump to the indices defined in \cref{section 3}. This correspondence will provide a new proof of Sturm's oscillation theorem. Appendices are devoted to technical results needed in \cref{section 3}. 

\section{Galilei covariance}\label{section 2}
The Galilei transformation (boost) 
\begin{equation}\label{boost}
\hat{s}=s\,,\qquad \hat{x}=x-vs   
\end{equation}
relates a frame $F$ conventionally at rest to a frame $\hat{F}$ which moves at constant velocity $v$.
The potential transforms as $\hat{V}(\hat{x},s)=V(x,s)$, i.e. as
\begin{equation}\label{transf potential}
\hat{V}(\hat{x},s)=V(\hat{x}+vs,s)\,.    
\end{equation}
For the Thouless theory to be applicable in the frame $\hat{F}$ too, $\hat{V}$ must also satisfy (\ref{periodicity1}, \ref{periodicity2}), though possibly with other periods $\hat{T}$ and $\hat{L}$, be they fundamental or not. To begin with, let the periods $T$ and $L$ be fundamental. We then have:
\begin{lemma}\label{lemma 0}
The potential $\hat{V}$ has spatial period $L$. It has temporal period $\hat{T}$ iff $v\hat{T}=mL$, $\hat{T}=nT$ for some $m,n\in\mathbb{Z}$.
\begin{proof}
By (\ref{transf potential}), $\hat{T}$ is a temporal period for $\hat{V}$ iff $(v\hat{T},\hat{T})$ is a spacetime translation leaving $V$ invariant, and hence a vector of the lattice generated by $(L,0)$ and $(0,T)$.
\end{proof}
\end{lemma}
In conclusion, $\hat{V}$ is admissible, provided that
\begin{equation}\label{gen case}
v\hat{T}=\hat{L}\,,    
\end{equation}
where
\begin{equation}\label{periods}
\hat{T}=nT\,, \qquad \hat{L}=mL    
\end{equation}
for some $n,m\in\mathbb{Z}$. Since (\ref{periodicity1}, \ref{periodicity2}) imply the same for $\hat{T}$ and $\hat{L}$, we may read \eqref{periods} as a redefinition of the periods and assume 
\begin{equation}\label{base case}
\hat{T}=T\,,\qquad \hat{L}=L\,,\qquad vT=L\,,    
\end{equation}
while no longer requiring them to be fundamental.\\

The covariance can be addressed from two points of view. The first one is (non-relativistic) electrodynamics. The charge and current densities transform as 
\begin{equation}\label{22a}
\hat{\rho}=\rho\,, \qquad \hat{\jmath}= j - \rho v\,,    
\end{equation}
where $\rho=\rho(x,s)$, $\hat{\rho}=\hat{\rho}(\hat{x},s)$ and likewise for $\hat{\jmath}$, $j$. 
The existing and transported charges are
\begin{equation*}
N\coloneqq\int_0^L\rho\,dx\,, \qquad Q\coloneqq\int_0^T j \,ds\,.    
\end{equation*}
They will be commented upon later, including their transformation law, which for $\rho$ and $j$ periodic reads
\begin{equation}\label{transf}
\hat{N}= mN\,, \qquad \hat{Q} = nQ - mN
\end{equation}
in the general case (\ref{gen case}), or just 
\begin{equation}\label{claim0}
\hat{N}= N\,, \qquad \hat{Q} = Q - N    
\end{equation}
in the simpler but not more restrictive case (\ref{base case}). The latter pair of equations implies the former because $N^{\prime}=mN$, $Q^{\prime}=nQ$, under the redefinition of periods, cf. (\ref{periods}). 
For now we observe that \cref{transf} is immediate from (\ref{22a}) in case $\rho$ and $j$ are constants, since then $N=\rho L$ and $Q=jT$.\\ 

The second point of view is to consider the bundle $\hat{P}$ as arising from the potential $\hat{V}$ in (\ref{transf potential}) by the same way $P$ arose from $V$ in (\ref{hamiltonian}).
\begin{theorem}[Galilei covariance]\label{theorem1}
Under the stated assumptions we have
\begin{equation}\label{24a}
\rk \hat{P} = m\rk P\,, \qquad \ch\hat{P} = n\ch P - m \rk P    
\end{equation}
in the general case (\ref{gen case}, \ref{periods}) and in particular
\begin{equation}\label{24b}
\rk  \hat{P}= \rk P\,, \qquad \ch\hat{P}= \ch P - \rk P
\end{equation}
in the simpler case $n=m=1$.
\end{theorem}
The matching transformation laws (\ref{transf}) and (\ref{24a}) embody the covariance of the theory of quantum pumps, as encapsulated in $Q=\ch P$, $N=\rk P$.\\

The proof of Theorem~\ref{theorem1} will be based on the transformation law between bundles, $P$ and $\hat{P}$. It  nonetheless pays to first forgo them and to use purely classical terms, in order to relate existing and transported charges, as they arise in different frames. We shall do so next.  

\begin{proof}[Proof of (\ref{transf})]
As mentioned, it suffices to prove the Eqs. (\ref{claim0}), the first of which follows from \cref{22a} for $\rho$. As for the second one, consider the shaded triangle in the diagram and its sides, see \cref{fig1}. 
\begin{figure}[H]
\centering
\begin{picture}(0,0)%
\includegraphics{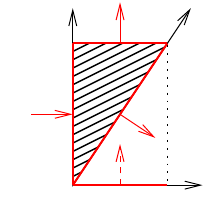}%
\end{picture}%
\setlength{\unitlength}{4144sp}%
\begingroup\makeatletter\ifx\SetFigFont\undefined%
\gdef\SetFigFont#1#2#3#4#5{%
  \reset@font\fontsize{#1}{#2pt}%
  \fontfamily{#3}\fontseries{#4}\fontshape{#5}%
  \selectfont}%
\fi\endgroup%
\begin{picture}(1557,1647)(-104,-1201)
\put(1396,299){\makebox(0,0)[lb]{\smash{{\SetFigFont{12}{14.4}{\familydefault}{\mddefault}{\updefault}{\color[rgb]{0,0,0}$\hat{x}=0$}%
}}}}
\put(240, 29){\makebox(0,0)[lb]{\smash{{\SetFigFont{12}{14.4}{\familydefault}{\mddefault}{\updefault}{\color[rgb]{0,0,0}$T$}%
}}}}
\put(1081,-1141){\makebox(0,0)[lb]{\smash{{\SetFigFont{12}{14.4}{\familydefault}{\mddefault}{\updefault}{\color[rgb]{0,0,0}$L$}%
}}}}
\put(300,-1141){\makebox(0,0)[lb]{\smash{{\SetFigFont{12}{14.4}{\familydefault}{\mddefault}{\updefault}{\color[rgb]{0,0,0}$0$}%
}}}}
\put(1396,-1141){\makebox(0,0)[lb]{\smash{{\SetFigFont{12}{14.4}{\familydefault}{\mddefault}{\updefault}{\color[rgb]{0,0,0}$x$}%
}}}}
\put(260,299){\makebox(0,0)[lb]{\smash{{\SetFigFont{12}{14.4}{\familydefault}{\mddefault}{\updefault}{\color[rgb]{0,0,0}$s$}%
}}}}
\put(170,-610){\makebox(0,0)[lb]{\smash{{\SetFigFont{12}{14.4}{\familydefault}{\mddefault}{\updefault}{\color[rgb]{1,0,0}$Q$}%
}}}}
\put(910,-440){\makebox(0,0)[lb]{\smash{{\SetFigFont{12}{14.4}{\familydefault}{\mddefault}{\updefault}{\color[rgb]{1,0,0}$\hat{Q}$}%
}}}}
\put(860,290){\makebox(0,0)[lb]{\smash{{\SetFigFont{12}{14.4}{\familydefault}{\mddefault}{\updefault}{\color[rgb]{1,0,0}$N$}%
}}}}
\end{picture}%
\caption{Existing and transported charges.}\label{fig1}
\end{figure}
Since the spacetime current density $\vec{\jmath}=(\rho, j)$ satisfies $\operatorname{div} \vec{\jmath}=0$, we find
\begin{equation*}
0 =\int_{\partial A} \vec{\jmath} \cdot d \vec{n} =\int_0^T(v \rho-j)\, ds\bigg|_{\hat{x}=0}+\int_0^T j \,ds\bigg|_{x=L}-\int_0^L \rho\, dx\bigg|_{s=0}\,,
\end{equation*}
where, in the last term, we moved the evaluation from $s=T$ to $s=0$ by the assumed periodicity. The first term equals $-\int_0^T \hat{\jmath}\,ds\big|_{\hat{x}=0}$. Thus $\hat{Q}=Q-N$.
\end{proof}
%
As an aside to \cref{theorem1}, Thouless pumps belong to class $A$ (no symmetry), $d=2$, of the periodic table of topological matter \cite{az,kitaev}. Accordingly, non-trivial indices arise in $d=2,0$. They are indeed $\ch P$ (strong) and $\rk P$ (weak).\\ 

Before engaging in the proof of \cref{theorem1}, we elaborate on the Fermi bundle $P$ seen there, as well as on the (larger) Bloch bundle $E$. To this end, let $\mathcal{T}_a$ be the translation operator, 
\begin{equation}\label{1}
(\mathcal{T}_a \psi)(x)=\psi(x-a)\,, \qquad (a \in \rr)\,,
\end{equation}
i.e. $\mathcal{T}_a=\e^{-\i pa}$ acting on $\psi \in$ $L^2(\rr, \cc^n)$ with generator $p=-\i \, d/dx$. More generally, the wavefunctions $\psi$ can be taken to be $L^2$ locally uniformly, as it appears fit for quasi-periodic wavefunctions $\psi$ of period $L$,
\begin{equation}\label{2}
\mathcal{T}_a \psi=\e^{-\i  ka} \psi\,, \qquad (a \in L \mathbb{Z})\,, 
\end{equation}
and quasi-momentum $k \in \rr\!\mod 2\pi/L \equiv S^1$. The Bloch bundle $E$ is a vector bundle \cite[Sect. 9.3]{nakahara}. Its base space is the torus $S^{1} \times S^{1} \ni(s, k)$, while its fibers $E_{s,k}\equiv E_k$ consist of quasi-periodic wavefunctions (\ref{2}). Since translations commute, they act fiberwise:
\begin{equation}\label{3}
\mathcal{T}_a: E_{s, k} \rightarrow E_{s, k}\,, \qquad(a\in \rr)\,;
\end{equation}
and so does
\begin{equation}\label{4}
p=-\i  \frac{d \mathcal{T}_a}{da}\bigg|_{a=0}\,.
\end{equation}
The fibers $E_{s,k}=E_k$ are equipped with an $L^2$-inner product obtained by restriction of $\psi$ to any quasi-period (of length $L$).
The main interest though is for the subbundle $P\subset E$. It is defined in terms of the spectral projection $P(s)$ associated to $H(s)$ and the interval $(-\infty, \mu]$ of energies below the Fermi level $\mu$ (Fermi projection). Basic properties are: $\mu$ lies in a spectral gap,
\begin{equation}\label{fermi level}
\mu\notin \sigma(H(s))\,, \qquad (s\in S^1)\,,    
\end{equation}
as by the standing assumption of the theory; $P(s)$ commutes with lattice translations, as implied by (\ref{periodicity2}),
\begin{equation}\label{5}
[\mathcal{T}_a, P(s)]=0\,, \qquad(a \in L\mathbb{Z})\,,
\end{equation}
and the dependence $s\mapsto P(s)$ is smooth thanks to the general assumptions. As a result, projections $P(s, k): E_{s, k}\rightarrow E_{s, k}$ of constant rank are induced, and with them a subbundle $P$ having fibers $P_{s,k}=\operatorname{ran} P(s, k)$, to the effect  that $\operatorname{ran} P(s)$ consists of $L^2$-sections of $P_{\{s\} \times S^{1}}$. As we shall see, $P$ is  of finite rank, $\rk P$. The position operator $x$ is of course not translation invariant,
\begin{equation}\label{6}
\mathcal{T}_a x \mathcal{T}_a^* = x - a\mathds{1}\,, \qquad (a \in \rr)\,,
\end{equation}
but the commutator $\i [P,x]$ is invariant for $a \in L\mathbb{Z}$, because of (\ref{5}) and $[P, \mathds{1}]=0$. As a result,
\begin{equation}\label{7}
\i [P, x]: E_{s, k} \rightarrow E_{s, k}
\end{equation}
acts fiberwise, and the same notation shall be kept for the so restricted action. The same conclusion can be reached as follows: The multiplication operator $\e^{\i \varphi x}$ maps $E_k \rightarrow E_{k+ \varphi}$, cf. (\ref{2}), hence
\begin{equation}\label{8}
\e^{-\i  \varphi x} P \e^{\i \varphi x}: E_k \rightarrow E_{k+\varphi} \rightarrow E_{k+\varphi} \rightarrow E_k
\end{equation}
(with index $s$ momentarily suppressed), and (\ref{7}) follows by
\begin{equation}\label{9}
\i [P, x]=\frac{d}{d\varphi}(\e^{-\i  \varphi x} P\,\e^{\i \varphi x})\Big|_{\varphi=0}\,.
\end{equation}
The Bloch trivialization of $E$ arises by identifying the fibers $E_{s,k}=E_k$ with the typical fiber $F$ consisting of $L^2$-periodic functions $u(x)$, and indeed $E_k\equiv F$, $\psi\mapsto u$ by means of ``Bloch's theorem''
\begin{equation}\label{10}
\psi(x)=u(x) \e^{\i kx}\,.
\end{equation}
The resulting trivialization of the (actually trivial) bundle $E$ is not global, since (\ref{10}) relies on a lift $\rr\rightarrow S^{1}$, $\widetilde{k} \mapsto k= \widetilde{k}\, \text{mod}\, 2\pi / L$ that is not invertible. We still denote by $P(s,k)$ the projection of the same name when viewed in the trivialization, i.e. as a map $F \rightarrow F$. Then, we claim,
\begin{equation*}
\i [P, x]=\partial_k P\,.
\end{equation*}
Indeed, $\e^{\i  \varphi x}: E_{s,k}\equiv F\rightarrow E_{s, k+\varphi}\equiv F$ acts as the identity, cf. (\ref{10}); thus (\ref{8}) acts as $P(s, k+\varphi)$ and the claim follows from (\ref{9}). We conclude that Eq. (\ref{chern}) is abstracted from the Bloch trivialization in which it was formulated by rewriting it as \cref{12}, where $\i [P,x]: E_{s,k}\rightarrow E_{s,k}$ is as in (\ref{7}) and the trace is over $E_{s, k}$.
\begin{proof}[Proof of \cref{theorem1}]
First of all, the Fermi projection $P(s)$ has fibers $P_{s,k}$ that are of finite rank, since the fibers of the Hamiltonian (\ref{hamiltonian}) have compact resolvent \cite[Thm.$\,$XIII.67]{reedsimoniv}. Second of all, we can set w.l.o.g. $m=n=1$. Indeed, we may consider the two redefinitions (\ref{periods}) one at a time. By replacing the period $L$ with $\hat{L}\coloneqq mL$, we obtain another one, cf. (\ref{periodicity2}). Under this change, the Brillouin zone shrinks by a factor $m$, in view of the identification of quasi-momenta differing by $2\pi/mL$, but their fibers get summed up (by way of a push-forward). More explicitly, denoting by $S_{\ell}^1$ a circle of length $\ell$, we have 
\begin{equation*}
S^1_{2\pi/L}\rightarrow S^1_{2\pi/\hat{L}}\,,\enskip k\mapsto k'=k\mod 2\pi/\hat{L}\,,\qquad P'_{s,k'}=\bigoplus_{k\mapsto k'}P_{s,k}\,.    
\end{equation*}
As a result, $\rk P'=m\rk P$, $\ch P'=\ch P$. Likewise, by replacing $T$ by $\hat{T}\coloneqq nT$, the period repeats $n$ times and the fibers get repeated along (by way of a pull-back), 
\begin{equation*}
S^1_{\hat{T}}\rightarrow S^1_{T}\,,\enskip s'\mapsto s=s'\mod T\,,\qquad P'_{s',k}=P_{s,k}\,,    
\end{equation*}
the result this time being $\rk P'=\rk P$, $\ch P'=n\ch P$. It remains to prove (\ref{24b}). We shall do so without resorting to the Bloch trivialization, and thus relying on (\ref{12}).\\

The potential in the moving frame $\hat{F}$ is (\ref{transf potential}), i.e.
\begin{equation}\label{13a}
\hat{V}=\mathcal{T}_{-vs} V \mathcal{T}_{vs}\,,
\end{equation}
where $\mathcal{T}_a$ is the translation (\ref{1}). The Fermi projections transform alike, $\hat{P}(s)=\mathcal{T}_{-vs}P(s) \mathcal{T}_{vs}$, and so do the projections acting on $E_{s,k}$:
\begin{equation}\label{13b}
\hat{P}(s, k)=\mathcal{T}_{-vs} P(s, k) \mathcal{T}_{vs}\,,
\end{equation}
cf. (\ref{3}). That in turn implies
\begin{align*}
& [\hat{P}, x]=\mathcal{T}_{-v s}[P, x] \mathcal{T}_{v s}\,, \\
& \partial_s \hat{P}=\mathcal{T}_{-vs}(\partial_s P-\i v[P,p]) \mathcal{T}_{vs} \,.
\end{align*}
The Chern numbers of the bundles $\hat{P}$ and $P$ determined by the two frames thus differ by
\begin{equation}\label{14}
\ch\hat{P}-\ch P=-\frac{\i v}{2 \pi} \oint ds \oint dk \tr(P[[P,x],[P, p]])\,,
\end{equation}
where we first got rid of $\mathcal{T}_{vs}$ by cyclicity when using \cref{12} on $\ch\hat{P}$. The trace is over $E_{s, k}$, still.
We next state a few identities:
\begin{equation}\label{15}
P[[P, x],[P, p]] P=[PxP, PpP]-P[x,p] P\,,
\end{equation}
with all three terms acting fiberwise on $E_{s,k}$;
\begin{align}
& \tr_{E_{s,k}}[PxP, PpP]=\i\partial_k \tr_{E_{s, k}}PpP\,; \label{16} \\
& \tr_{E_{s,k}} P[x, p]P=\i \tr_{E_{s, k}} P\,. \label{17}
\end{align}
Before proving the identities, we apply them to (\ref{14}), so as to obtain
$$
\ch \hat{P} - \ch P=-\frac{v}{2 \pi} \oint ds \oint dk \rk  P= -\rk  P\,.
$$
Indeed, for the first equality, we used that (\ref{16}) is a total derivative and that $\tr_{E_{s,k}}P=\rk P$ in (\ref{17}); for the second one, that 
\begin{equation}\label{periods ratio}
(v/ 2\pi) T (2\pi /L)=1\,,   
\end{equation} 
by (\ref{base case}). That proves the second equation (\ref{24a}). It remains to show the identities (\ref{15}-\ref{17}) above. The first one follows from
$$
P[P,x][P, p] P=PxPpP- PxpP
$$
and in turn from $P[P,x]=P x(1-P)$, $[P, p]P=-(1-P)pP$. The l.h.s. of (\ref{15}) is seen to act fiberwise by (\ref{4}, \ref{7}); the second term on the r.h.s. does too by $[x, p]= -\i  $, which by the way proves (\ref{17}). Hence, the first term also acts fiberwise, which could be seen independently along the lines used to show (\ref{7}). Just as $x$ did not fiber in that context, neither does $PxP$ here, but $[PpP,PxP]$ will. Indeed, by (\ref{5}, \ref{6}) we have
$$
\mathcal{T}_a \e^{\i  \varphi PxP}=\e^{\i  \varphi P(x-a) P}\mathcal{T}_a=\e^{-\i  \varphi aP} \e^{\i  \varphi PxP} \mathcal{T}_a\,, \qquad(a \in L\mathbb{Z})\,.
$$
That shows that $\e^{\i \varphi PxP}: P_{s,k} \rightarrow P_{s, k+\varphi}$ and
\begin{equation*}
A(\varphi)\coloneqq \e^{-\i \varphi PxP} PpP \e^{\i \varphi PxP}: P_{s,k} \rightarrow P_{s,k}    
\end{equation*}
and even $E_{s,k} \rightarrow E_{s,k}$, since the action on $P_{s, k}^{\perp} \subset E_{s, k}$ is annihilation. Finally, the same holds for $[PpP, PxP]$, with
\begin{equation*}
\i  \tr_{E_{s,k}}[PpP, PxP]=\frac{d}{d\varphi} \tr_{P_{s,k}} A(\varphi)\bigg|_{\varphi=0}=\frac{d}{d \varphi} \tr_{P_{s, k+\varphi}}(PpP)\bigg|_{\varphi=0}=\frac{\partial}{\partial k} \tr_{P_{s,k}}(PpP)\,.
\end{equation*}
\end{proof}
\begin{remark}[Gauss-Codazzi]
Let $E\rightarrow M$ be a Hermitian vector bundle; let $P\subset E$ denote a subbundle, as well as the corresponding orthogonal projection; let $\nabla$ be a Hermitian connection on $E$ and let $\nabla^{(P)}=P\nabla P$ be the induced connection; let $A$ be the second fundamental form, $\nabla=\nabla^{(P)}+A$, i.e.
\begin{equation*}
A=P^{\perp}\nabla P\,;    
\end{equation*}
let finally $R$ and $R^{(P)}$ be the curvatures of $\nabla$, $\nabla^{(P)}$, respectively. They are related by the Gauss-Codazzi equation, cf. \cite[Eq.$\,$1.6.12]{kobayashi}, 
\begin{equation}\label{gauss codazzi}
PRP=R^{(P)}-A^*\wedge A\,.    
\end{equation}
A similar construction, but devoid of base space $M$, is as follows. Let $\mathcal{H}$ be a Hilbert space and
\begin{equation*}
\mathcal{H}\rightarrow\mathcal{H}\,,\qquad \psi\mapsto e^{Xt}\psi\,,\qquad (t\in\rr)   
\end{equation*}
be the ``parallel transport'' generated by the anti-Hermitian operator $X=-X^*$. Thus $\nabla$ is given by $\nabla_X\psi=X\psi$. Let the orthogonal projection $P$ define a subspace of $\mathcal{H}$, also denoted $P$. Then the induced transport is generated by $\nabla^{(P)}$ with
\begin{equation*}
\begin{gathered}
\nabla_X=\nabla_X^{(P)}+A(X)\,,\\
A(X):P\rightarrow P^{\perp}\,,\qquad A(X)= P^{\perp}XP= [X,P]P\,.    
\end{gathered}
\end{equation*}
We have
\begin{align*}
& A^*(X): P^{\perp}\rightarrow P\,, \qquad A^*(X)=P[X,P]\,, \\
& R(X, Y)=\nabla_X \nabla_Y-\nabla_Y \nabla_X-\nabla_{[X, Y]}=0\,, \\
& R^{(P)}(X, Y)=[P X P, P Y P]-P[X, Y] P
\end{align*}
and (\ref{gauss codazzi}) reads
\begin{equation*}
[P X P, P Y P]-P[X, Y] P=P[[P, X],[P, Y]] P\,.
\end{equation*}
For $\mathcal{H}=L^2(\rr)$, $X=\i x$, $Y=\i p$, this is (\ref{15}).
\end{remark}
We next present an alternative proof of \cref{theorem1} that makes use of the Bloch trivialization, rather than foregoing it. (For a broad discussion of that trivialization we refer to \cite{gawedzki}.) We recall that it is not global on $\mathbb{T}=S_s^{1} \times S_k^{1}$, but rather on its lift $S_s^{1} \times \rr_k$. Let $G=2 \pi /L$ be the primitive dual lattice vector. Fibers $E_k \cong E_{k+G}$ are both represented by $F$, though differently, and in fact by $F \rightarrow F$,
$$
u_{k+G}(x)=\e^{-\i  G x} u_k(x)\,,
$$
cf. (\ref{10}). Hence
\begin{equation*}
P(s, k +G)=\e^{-\i Gx} P(s, k) \e^{\i Gx}\,,
\end{equation*}
where the conjugating unitary $\e^{-\i Gx}$ is independent of $k$. In passing, we remark that the integrand of (\ref{chern}) is nonetheless defined globally on $\mathbb{T}$, precisely due to this independence. The translation operator $\mathcal{T}_a$, cf. (\ref{1}), is given in the present trivialization by 
\begin{equation}\label{alt translations}
\widetilde{\mathcal{T}}_a : F \rightarrow F\,, \qquad \widetilde{\mathcal{T}}_a=\mathcal{T}_a \e^{-\i ka}\,, \qquad (a\in\mathbb{R})\,,
\end{equation} 
as seen from (\ref{10}) and $(\widetilde{\mathcal{T}}_a u)(x) =$ $u(x-a) \e^{-\i ka}$. We note that even lattice translations, $\widetilde{\mathcal{T}}_a=e^{-\i ka}$, $(a\in L\mathbb{Z})$, act differently on fibers $E_k$ and $E_{k+G}$. The generator of $\widetilde{\mathcal{T}}_a$ is $\widetilde{p}=p+k$ and it satisfies
\begin{equation}\label{19}
p+k+G = \e^{-\i Gx}(p+k) \e^{\i Gx}\,.    
\end{equation}
As a final preliminary, we recall \cite[Lm. 3.2]{cmp/1104180308} (see also \cite[Eq. (8)]{gawedzki}). 
\begin{lemma}
Let $P$ be a function on $\rr^n$ with values in the projections of finite rank in some Hilbert space. Let $U$ be a function on $\rr^n$ with values in the invertible linear operators on this Hilbert space. Assuming $P$ and $U$ to be (strongly) differentiable, the following identity holds:
\begin{equation}\label{20}
\tr P_U dP_U\wedge dP_U - \tr P dP\wedge dP = d(\tr P U^{-1}dU)  
\end{equation}
where $P_U\coloneqq U P U^{-1}$.
\end{lemma}
\begin{proof}[Proof of \cref{theorem1} (alternative)]
The proof is placed in the context of the Bloch trivialization. Eq. (\ref{13a}) holds true with $\widetilde{T}$ in place of $T$, cf. (\ref{alt translations}), and so does (\ref{13b}), with the proviso that $P$ and $\hat{P}$ are now maps on $F$. We then use (\ref{20}) on $\rr^2 \ni(s, k)$ and apply it to $U=\widetilde{\mathcal{T}}_{-vs}$, so that the l.h.s. (times $\i /2\pi$) is the integrand of $\ch\hat{P}-\ch P$. The r.h.s. does not descend to an exact 1-form on the torus $\mathbb{T}$, so we cannot conclude that the integral thereof vanishes. Instead, we will apply Stokes' theorem to the rectangle $R=[0, T] \times[0, G] \ni(s,k)$. By $d \widetilde{\mathcal{T}}_a=-\i (p da + d(ka)) \widetilde{\mathcal{T}}_a$ we find
\begin{equation*}
U^{-1} dU=\widetilde{\mathcal{T}}_{-vs}^* d\widetilde{\mathcal{T}}_{-vs}=\i v((p+k) ds +s dk)\,.
\end{equation*}
In view of (\ref{13b}, \ref{19}), the contributions coming from $ds$ and the segments $k=0, G$ of $\partial R$ vanish. We are left with
\begin{equation*}
\ch \hat{P}-\ch P=- \frac{v}{2\pi} \int_0^G dk \,s \tr P(s,k)\Big|_{s=0}^{s=T}= - \rk   P\,,
\end{equation*}
again by (\ref{periods ratio}).
\end{proof}
\begin{remark}
In (\ref{20}), $\operatorname{tr}(PU^{-1}dU)$ is the Chern-Simons 1-form of $\operatorname{tr}(P_U dP_U\wedge dP_U)$ in a local chart over which $P$ is trivialized \cite[Sect. 11.5]{nakahara}.
\end{remark}
\section{Analysis in position space}\label{section 3}
The Hamiltonian acting on $L^2(\rr_x , \cc^n)$ is 
\begin{equation}\label{hamiltonian pos}
H(s)=-\frac{d^2}{dx^2}+V(x,s)\,, \qquad (s\in S^1)\,,
\end{equation} 
where the potential $V(x,s)=V(x,s)^{*}$ is parametrized by the adiabatic time $s=\varepsilon t$, with $\varepsilon\ll 1$. Let $V(x,s)$ be periodic in time and space with periods $T$ and $L$, cf. (\ref{periodicity1}, \ref{periodicity2}). We shall next present a formulation of Galilei covariance that does not rely on Bloch theory. It rather rests on a framework \cite{ortelli}, to be reviewed below, that in itself requires time but not space periodicity of the potential $V(x,s)$, besides the gap condition (\ref{fermi level}). The analysis actually relies on position space. In the present context, time periodicity must apply in both frames, with periods $T$ and $\hat{T}$. 

We assume $T$ and $\hat{T}$ to be commensurate.
\begin{lemma}
Then $V$ is spatially periodic. In particular, (\ref{base case}) may be assumed.
\begin{proof}
As remarked in the proof of \cref{lemma 0}, the translation $(v\hat{T},\hat{T})$ is a spacetime symmetry of $V$. By $mT=m'\hat{T}\eqqcolon \widetilde{T}$ we see that $m'(v\hat{T},\hat{T})-m(0,T)=(v\widetilde{T},0)$ is a spatial period.
\end{proof}
\end{lemma}

In the setting (\ref{hamiltonian pos}), at first we fix $s\in S^1$ and $z\in\cc$ outside of the spectrum, $z\notin\sigma(H(s))$. We then consider solutions $\psi=\psi(x)$, $\psi(x)\in\cc^n$ (column vectors) of the Schrödinger equation
\begin{equation}\label{schroed1}
-\psi^{\prime \prime}(x)+V(x, s) \psi(x)=z \psi(x)\,.   
\end{equation}
For any $x_0\in\rr$, any initial values $\Psi_0=(\psi_0,\psi'_0)\in\cc^n\oplus\cc^n$ uniquely determine a solution $\psi(x)$ of (\ref{schroed1}) having $\psi(x_0)=\psi_0$, $\psi'(x_0)=\psi'_0$. We denote that by
\begin{equation}\label{initial cond}
(\Psi_0,x_0)\leftrightarrow \psi\,.    
\end{equation}
We also consider the adjoint equation
\begin{equation}\label{schroed2}
-\widetilde{\psi}^{\prime \prime}(x)+\widetilde{\psi}(x) V(x, s)=z \widetilde{\psi}(x)
\end{equation}
and its solutions $\widetilde{\psi}=\widetilde{\psi}(x)$, $\widetilde{\psi}(x)\in\cc^n$ (row vectors). Likewise, we write
\begin{equation*}
(\widetilde{\Psi}_0,x_0)\leftrightarrow \widetilde{\psi}\,.    
\end{equation*}
For any (differentiable) functions $\psi(x)$, $\widetilde{\psi}(x)$, we have their Wronskian
\begin{equation}\label{wronsk}
W(\widetilde{\psi}, \psi ; x)\coloneqq\widetilde{\psi}(x) \psi^{\prime}(x)-\widetilde{\psi}^{\prime}(x) \psi(x)\,,    
\end{equation}
which defines a function $\rr\rightarrow\cc$, $x\mapsto W(\widetilde{\psi},\psi;x)$. If $\psi,\widetilde{\psi}$ are solutions of (\ref{schroed1}, \ref{schroed2}), then the function is constant and simply denoted $W(\widetilde{\psi},\psi)$. We emphasize the symplectic form on $\cc^n\oplus\cc^n$,
\begin{equation}\label{bilinear form}
\omega(\widetilde{\Psi},\Psi)\coloneqq\widetilde{\psi}\psi'-\widetilde{\psi}'\psi\,, \qquad (\Psi=(\psi,\psi')\,, \enskip \widetilde{\Psi}=(\widetilde{\psi},{\widetilde{\psi}}^{\prime}))\,.    
\end{equation}
We remark that $\omega$ is a non-degenerate bilinear form $\cc^{2n}\times\cc^{2n}\rightarrow\cc$. Moreover, \cref{wronsk} is restated as
\begin{equation}\label{wronsk as bilinear}
W(\widetilde{\psi},\psi;x)=\omega(\widetilde{\Psi}(x),\Psi(x))   
\end{equation}
for any functions $\psi$, $\widetilde{\psi}$ and $\Psi(x)=(\psi(x),\psi'(x))$, $\widetilde{\Psi}(x)=(\widetilde{\psi}(x),\widetilde{\psi}'(x))$. In particular,
\begin{equation}\label{3.7'}
W(\widetilde{\psi},\psi)=\omega(\widetilde{\Psi},\Psi)\,,    
\end{equation}
i.e. the Wronskian of two  solutions of the paired Schrödinger equations is the symplectic form evaluated on their initial values at any point.\\  

The following Lemma is about solutions that decay at one end of the line $\rr_x$.
\begin{lemma}\label{lemma1}
Let 
\begin{align}
\label{lemma1 eq1}
V^+_{(z,s,x_0)} &=\{\Psi\in\cc^n\oplus\cc^n\mid \psi \text{ as in (\ref{schroed1}) is $L^2$ at $x=+\infty$}\}\,,\\
\nonumber
\widetilde{V}^-_{(z,s,x_0)} &=\{\widetilde{\Psi}\in\cc^n\oplus\cc^n\mid \widetilde{\psi} \text{ as in (\ref{schroed2}) is $L^2$ at $x=-\infty$}\}\,.
\end{align}
Then
\begin{enumerate}[label=(\roman*), wide]
\item $V^+_{(z,s,x_0)}$, $\widetilde{V}^-_{(z,s,x_0)}$ are unchanged under $x_0\rightarrow x_0+L$.
\item $\dim V^+_{(z,s,x_0)}=n$, $\dim \widetilde{V}^-_{(z,s,x_0)} = n$.
\item The bilinear form $\omega$, cf. (\ref{bilinear form}), remains non-degenerate upon restriction to $\widetilde{V}^-_{(z,s,x_0)}\times V^+_{(z,s,x_0)}$.
\end{enumerate}
\end{lemma}
The proof is given in \cref{app b}.\\

Some vector bundles now arise naturally. Let first
\begin{equation}\label{3torus}
\mathbb{T}_3=\gamma\times (\rr/T\mathbb{Z})\times(\rr/L\mathbb{Z})\ni (z,s,x_0)    
\end{equation}
be the 3-torus obtained from: 
\begin{itemize}[leftmargin=*]
\item A contour $\gamma\subset\cc$ encircling the part of the spectrum $\sigma(H(s))$ that is below the Fermi level, for all $s$;
\item A circle $\rr/T\mathbb{Z}\ni s$ for time
\item A circle $\rr/L\mathbb{Z}\ni x$ for space.
\end{itemize}
\begin{figure}[H]
\centering
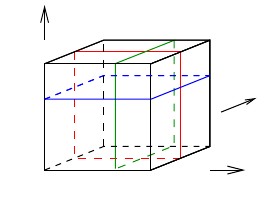
\caption{ The 3-torus $\mathbb{T}_3$, cf. (\ref{3torus}), and its 2-dimensional
subtori $\mathbb{T}_x$, $\mathbb{T}_s$, and $\mathbb{T}_z$, obtained by
fixing the value of one of the three coordinates.}
\label{figure 3torus}
\end{figure}
A preliminary and trivial vector bundle is $\mathbb{T}_3\times(\cc^n\oplus\cc^n)$ and it comes into two copies, denoted $E$ and $\widetilde{E}$, in which $\cc^n$ is viewed as consisting of column and row vectors, respectively. Given a potential $V=V(x,s)$ as in (\ref{hamiltonian pos}), the fiber of $E$ at $(z,s,x_0)$ is identified by (\ref{initial cond}) with the vector space of initial data. Two vector bundles 
\begin{equation}\label{3.10a}
P\subset E\,,\qquad \widetilde{P}\subset\widetilde{E}
\end{equation}
are then introduced as subbundles having fibers $V^+_{(z,s,x_0)}$ and $\widetilde{V}^-_{(z,s,x_0)}$, respectively. They are well-defined by (i) of \cref{lemma1}.\\

The bundle $E$ may be restricted to the 2-torus $\mathbb{T}_x=\gamma\times S^1\times \{x\}$, obtained from $\mathbb{T}$ by fixing $x$ to any value, and along that also the subbundle $P$, resulting in $P_x$. Likewise, the restriction to $\mathbb{T}_s=\gamma\times\{s\}\times\rr/L\mathbb{Z}$, for any fixed $s$, yields $P_s$ (no notational confusion between $P_x$ and $P_s$ will occur). We retain
\begin{equation}\label{3.8a}
P_x=P\mid_{\mathbb{T}_x}\,, \qquad P_s=P\mid_{\mathbb{T}_s}\,.    
\end{equation}
For definiteness, $\mathbb{T}_{x}$ and $\mathbb{T}_{s}$ shall be given the orientations $dz \wedge ds$ and $dx \wedge dz$, respectively. Clearly, the Chern numbers of the two bundles are independent of $x$ and $s$. 
\begin{proposition}[Charges]\label{prop2}
\begin{enumerate}[label=(\roman*), wide]
\item The charge transport during a temporal period is
\begin{equation*}
Q=\ch P_x\,, \qquad (\text{any }x)\,.    
\end{equation*}
\item The charge residing in a spatial period is
\begin{equation*}
N=\ch P_s\,, \qquad (\text{any }s)\,.    
\end{equation*}
\end{enumerate}
\end{proposition}
\begin{proof}
We shall apply \cref{prop appA} to the bundles $P_x$ and $P_s$, cf. (\ref{3.8a}), in each case together with bundles $\widetilde{P}_x$ and $\widetilde{P}_s$ likewise obtained from $\widetilde{P}$. In both cases the bilinear will be the symplectic form (\ref{bilinear form}). Biorthogonal frames $(\widetilde{F},F)$, as in (\ref{appA 4}, \ref{appA 5}), are obtained for $\widetilde{P},P$ as follows: A (local) basis for $P_p$, $p=(z,s,x)$ consists of $n$ linearly independent (column) solutions of (\ref{schroed1}) as in (\ref{lemma1 eq1}) --- collectively one matrix solution $\psi=\psi(x)\in M_n(\cc)$ --- or rather of their initial values at $x$. That is: Data $\Psi=(\psi,\psi')\in M_n(\cc)\oplus M_n(\cc)$ such that 
\begin{itemize}[leftmargin=*]
    \item the corresponding solution $\psi(x)$ is in $L^2$ at $x=+\infty$; 
    \item $\psi a=0, \enskip \psi' a=0 \enskip \Longrightarrow \enskip a=0$.
\end{itemize}
Likewise, a basis $\widetilde{F}$ for $\widetilde{P}_p$ is obtained from $\widetilde{\Psi}=(\widetilde{\psi},\widetilde{\psi}')$. Then (\ref{appA 5}) says
\begin{equation}\label{gauge fix}
\omega(\widetilde{\Psi},\Psi)=\mathds{1}\,,    
\end{equation}
whereby (\ref{wronsk}) has become matrix-valued with entries $\omega(\widetilde{\Psi},\Psi)_{ij}=\omega(\widetilde{\Psi}_i,\Psi_j)$. \cref{appA 7} then leads to
\begin{equation}\label{chern2}
\ch P_x =\frac{\mathrm{i}}{2 \pi}\int_{\gamma\times (\rr/ T\mathbb{Z})} dz\,ds  \tr\Bigr(W\Bigr(\frac{\partial \widetilde{\psi}}{\partial z}, \frac{\partial \psi}{\partial s} ; x\Big)-W\Bigr(\frac{\partial \widetilde{\psi}}{\partial s}, \frac{\partial \psi}{\partial z} ; x\Big)\Big)\,,
\end{equation}
\begin{equation}\label{chern3}
\ch P_s =\frac{\mathrm{i}}{2 \pi}\int_{ (\rr/ L\mathbb{Z})\times\gamma} dx\,dz  \tr\Bigl(W\Bigl(\frac{\partial \widetilde{\psi}}{\partial x}, \frac{\partial \psi}{\partial z} ; x\Bigr)-W\Bigl(\frac{\partial \widetilde{\psi}}{\partial z}, \frac{\partial \psi}{\partial x} ; x\Bigr)\Bigr)\,.
\end{equation}
Indeed, the computation proceeds along (\ref{a9}) with $b=\omega$. In particular, the derivatives seen there are carried over as derivatives in $W(\tilde\psi,\psi;x)$ acting on $\psi$, $\tilde\psi$, but not on $x$, cf. (\ref{wronsk as bilinear}). We observe that the r.h.s. of (\ref{chern2}) has been identified to be equal to $Q$ in \cite[Thm. 1]{ortelli}, but (\ref{chern3}) has not been identified with $N$. We will do so next. More precisely, we claim that, for any $x$, $s$, 
\begin{equation}\label{charge density}
\frac{\mathrm{i}}{2 \pi}\int_{\gamma} dz  \tr\Bigl(W\Bigl(\frac{\partial \widetilde{\psi}}{\partial x}, \frac{\partial \psi}{\partial z} ; x\Bigr)-W\Bigl(\frac{\partial \widetilde{\psi}}{\partial z}, \frac{\partial \psi}{\partial x} ; x\Bigr)\Bigr)=\tr P(x,x)\,,    
\end{equation}
where $P=P(s)$ is the Fermi projection, i.e. the spectral projection in Hilbert space $L^2(\rr_x,\cc^n)$ associated to the part of the spectrum of $H(s)$ encircled by $\gamma$:
\begin{equation*}
P(s)=\frac{\mathrm{i}}{2\pi}\int_{\gamma} dz(H(s)-z)^{-1}\,.    
\end{equation*}
Thus $\tr P(x,x)$ is the charge density at $x$. In particular, (\ref{chern3}, \ref{charge density}) imply
\begin{equation*}
\ch P_s=\int_{\rr/L\mathbb{Z}}dx \tr P(x,x)=N\,,    
\end{equation*}
as claimed.
\end{proof}
\begin{proof}[Proof of (\ref{charge density})]
As a preliminary we note that
\begin{equation}\label{prelim}
\frac{\partial \widetilde{\psi}}{\partial z} \psi''=\frac{\partial \widetilde{\psi}''}{\partial z} \psi+\widetilde{\psi} \psi\,;
\end{equation}
indeed, this follows from (\ref{schroed1}, \ref{schroed2}), i.e.
\begin{equation*}
\psi''=(V-z) \psi\,,\qquad \widetilde{\psi}''=\widetilde{\psi}(V-z)\,,
\end{equation*}
and the derivative w.r.t. $z$ of the latter,
\begin{equation*}
\frac{\partial \widetilde{\psi}''}{\partial z} =\frac{\partial \widetilde{\psi}}{\partial z} (V-z)-\widetilde{\psi}\,.
\end{equation*}
We next expand the expression under the trace in the l.h.s. of (\ref{charge density}). It so equals
\begin{equation*}
\widetilde{\psi}' \frac{\partial \psi'}{\partial z}-\widetilde{\psi}'' \frac{\partial \psi}{\partial z}-\frac{\partial \widetilde{\psi}}{\partial z} \psi''+\frac{\partial \widetilde{\psi}'}{\partial z} \psi' =-\frac{\partial}{\partial z}(\widetilde{\psi}^{\prime \prime} \psi)+\frac{\partial}{\partial z}(\widetilde{\psi}' \psi')-\widetilde{\psi} \psi\,,
\end{equation*}
where we used (\ref{prelim}). Since any complex vector bundle on a circle is trivial, so is the bundle $P$ restricted to $\gamma\times\{(s,x)\}\equiv\gamma$. Hence, the frames $F=\Psi=(\psi,\psi')$ and $\widetilde{F}$ can be chosen to be global on $\gamma$. That implies for the l.h.s. of (\ref{charge density})
\begin{equation}\label{prelim2}
\frac{\mathrm{i}}{2 \pi}\int_{\gamma} dz \tr\Bigl(W\Bigl(\frac{\partial \widetilde{\psi}}{\partial x}, \frac{\partial \psi}{\partial z} ; x\Bigr)-W\Bigl(\frac{\partial \widetilde{\psi}}{\partial z}, \frac{\partial \psi}{\partial x} ; x\Bigr)\Bigr)= -\frac{\mathrm{i}}{2 \pi}\int_{\gamma}dz\operatorname{tr} \widetilde{\psi}(x)\psi(x)\,.    
\end{equation}
It will be shown in \cref{app C} that the Green's function, i.e. the integral kernel of the resolvent $(H(s)-z)^{-1}$, is given by
\begin{equation}\label{green}
G(x,x';z)=-\psi(x)\widetilde{\psi}(x')\,, \qquad (x\geq x')\,.    
\end{equation}
Thus (\ref{prelim2}) equals
\begin{equation*}
\frac{1}{2 \pi \mathrm{i}} \int_\gamma \tr G(x, x)dz =\tr P(x,x)   
\end{equation*}
by cyclicity of the trace. That proves (\ref{charge density}).
\end{proof}
\begin{theorem}[Galilei covariance]\label{theorem2}
Suppose two frames are related by (\ref{boost}, \ref{periods}). Then
\begin{align}
\label{main1 pos}
\ch\Hat{P}_{\Hat{x}} &=n \ch P_x -m\ch P_s\,,\\
\label{main2 pos}
\ch\Hat{P}_{\Hat{s}} &=m \ch P_s
\end{align}
for any $(x,s)$, $(\Hat{x},\Hat{s})$. Here $P,\Hat{P}\subset E$ are the vector bundles introduced following \cref{lemma1} and based on $V,$ $\Hat{V},$ cf. (\ref{transf potential}).
\end{theorem}
\begin{remark}
The theorem underpins \cref{transf} in view of \cref{prop2} and it is analogous to \cref{theorem1}, while foregoing Bloch theorem.
\end{remark}
\begin{proof}
As in other approaches, it suffices to consider $n=m=1$. The Galilei transformation acts on wavefunctions as
\begin{equation*}
\mathcal{T}_{-vs}:\psi(x)\mapsto\Hat{\psi}(\Hat{x})=\psi(x)=\psi(\Hat{x}+vs)\,,    
\end{equation*}
and likewise on $\psi'(x)$. In terms of the (trivial) bundle $E$, it acts as the identity between copies of the typical fiber $\cc^n\oplus\cc^n$ placed at $x_0$ and $\Hat{x}_0=x_0-vs$. In terms of the subbundles $P$ and $\Hat{P}$, the map appropriately acts between fibers $\smash{V^+_{(z,s,x_0)}\rightarrow \Hat{V}^+_{(z,s,x_0-vs)}}$. In terms of the frames seen in (\ref{chern2}) when that equation is applied to $\ch\Hat{P}_{\Hat{x}}$, we have 
\begin{equation*}
\Hat{\psi}(z,s,\Hat{x})=\psi(z,s,\Hat{x}+vs)\,,    
\end{equation*}
and thus 
\begin{equation}\label{derivatives new frame}
\frac{\partial\Hat{\psi}}{\partial s}=\frac{\partial\psi}{\partial s} +v \frac{\partial\psi}{\partial x}\,, \qquad \frac{\partial\hat{\psi}}{\partial z}=\frac{\partial\psi}{\partial z}\,,   
\end{equation}
and likewise for $\Hat{\widetilde{\psi}}$. Hence the trace in (\ref{chern2}), or rather the first term under it, becomes
\begin{equation*}
W\Bigl(\frac{\partial\Hat{\widetilde{\psi}}}{\partial z},\frac{\partial\Hat{\psi}}{\partial s};\Hat{x}\Bigr)= W\Bigl(\frac{\partial\widetilde{\psi}}{\partial z},\frac{\partial\psi}{\partial s};x\Bigr) +v W\Bigl(\frac{\partial\widetilde{\psi}}{\partial z},\frac{\partial\psi}{\partial x};x\Bigr)\,.    
\end{equation*}
Since $\ch P_x$, $\ch\Hat{P}_{\Hat{x}}$ are independent of $x$ and $\Hat{x}$ (cf. \cref{prop2}), we may average them over $\rr/L\mathbb{Z}$ without harm. We obtain
\begin{align*}
\ch\Hat{P}_{\Hat{x}}-\ch P_x &=\frac{\mathrm{i}}{2\pi}\int_{\gamma\times (\rr/T\mathbb{Z})} dz\, ds \frac{v}{L}\int_{\rr/L\mathbb{Z}} dx\tr\Bigl(W\Bigl(\frac{\partial\widetilde{\psi}}{\partial z},\frac{\partial\psi}{\partial x};x\Bigr) - W\Bigl(\frac{\partial\widetilde{\psi}}{\partial x},\frac{\partial\psi}{\partial z};x\Bigr) \Bigr)\\ &=-\frac{v}{L}\int_{\rr/T\mathbb{Z}}ds \ch P_s =-\ch P_s
\end{align*}
by (\ref{chern3}), the independence of $P_s$ on $s$ and $vT/L=1$. That proves (\ref{main1 pos}). \cref{main2 pos} follows more simply from (\ref{chern3}), the second \cref{derivatives new frame} and
\begin{equation*}
\frac{\partial\Hat{\psi}}{\partial\Hat{x}}=\frac{\partial\psi}{\partial x}\,.    
\end{equation*}
\end{proof}
Along with $\mathbb{T}_x$ and $\mathbb{T}_s$, a third 2-torus $\mathbb{T}_z=\{z\}\times S^1\times S^1$ is obtained from $\mathbb{T}_3$ by fixing $z\in\gamma$ to any value, cf. \cref{figure 3torus}. Let $\mathbb{T}_z$ be given the orientation $ds \wedge dx$. Like in (\ref{3.8a}), the restriction $P_z=P|_{\mathbb{T}_z}$ comes up naturally. Though it plays no role in the main result \cref{theorem2}, it pays to know it is trivial.
\begin{proposition}\label{prop 3.5}
\begin{equation}\label{third chern number}
\ch P_z=0\,,\qquad (z\in\gamma)\,.
\end{equation}
\end{proposition}
The proof will be given at the end of \cref{app b}.

\section{The scattering perspective}\label{section 4}
In this section we shall review the bulk-edge correspondence \cite{ortelli} for Thouless pumps, which relates their topological data $\ch P$ and $\rk P$ with scattering data. The latter data are at the core of the Büttiker theory \cite{buettiker,Citro} of quantum pumps, the scope of which is wider, in that pumps may be non-topological. The setup of that theory is that of a pump proper connected to any number of (semi-infinite) 1D channels, labelled $\alpha$ and sharing a common Fermi level $\mu$. The pump is configured in terms of some external parameter $X$, that determine a scattering matrix $S=(S_{\alpha\beta})=S(X)$ at energy $\mu$. As $X$ changes slowly in time, a (net) current flows through the channels. Upon completing a cycle in $X$ (or just in $S$), a net charge $Q_\alpha =(\i/2\pi)\oint ((dS)S^*)_{\alpha\alpha}$ has left the pump through channel $\alpha$.\\

If the pump proper has a gap at the Fermi level, then $S$ becomes block-diagonal w.r.t. groups $A$ of
channels $\alpha$ docking at the same point, provided different ones are widely separated. The block of any group stands for its (unitary) reflection matrix $R_A(X)$. Despite transmission being suppressed at fixed parameter $X$, transport remains possible when $X$ is varied. Indeed, $Q_A=\sum_{\alpha\in A} Q_\alpha$ equals the winding number of $\det R_A$.\\

Given a Thouless pump, a Büttiker pump is obtained by replacing the portion $x\le x_0$ of the former by an ideal channel, say one hosting free electrons. Cycles in $R(X)$ are then completed when either the truncation point $x_0$ or the time $s$ undergo a spatial and temporal period, respectively. The
corresponding bulk-edge correspondences will be stated rigorously in \cref{prop 4.2}. As a consequence, it will suffice to establish Galilei covariance in Büttiker's theory.  \\

The precise context is that of \cref{section 3}. We recall in particular the bundle $P$ having a 3-torus $\mathbb{T}_{3}$ as base space, as well as its restrictions (\ref{3.8a}) to 2-tori that arise by fixing $x$ or $s$. We consider the truncation of the potential $V(x, s)$ to $x \geq x_{0}$, whereby (\ref{hamiltonian}) gets replaced by
\begin{equation}\label{hamiltonian edge}
H(s, x_{0}) = p^{2} + \theta(x - x_{0}) V(x, s)\,,
\end{equation}
with $\theta$ being the Heaviside function. We shall suppress the variable $s$ from the notation until further notice.\\

Incoming (R) and outgoing (L) plane waves
\begin{equation}\label{scattering waves}
\psi_{R / L} \e^{ \pm \i k(x-x_{0})}\,, \qquad (k > 0\,,\, \psi_{R / L} \in \mathbb{C}^{n})    
\end{equation}
are now solutions of the time-independent Schrödinger equation for (\ref{hamiltonian edge}), for $x < x_{0}$ and at real energies $z = k^{2}$. We consider them and their superpositions at the energy
$$
z = \mu\,.
$$
By the gap condition (\ref{fermi level}) the matching solution for $x > x_{0}$ is decaying at $x \rightarrow +\infty$, thus preventing transmission. The reflection matrix,
$$
\psi_{L} = R(x_{0}) \psi_{R}\,,
$$
is unitary in $M_{n}(\mathbb{C})$ and given by \cite[Eq. (45)]{ortelli}
\begin{equation}\label{refl coeff}
R(x_{0}) = -(\psi^{\prime}(x_{0}) + \i k \psi(x_{0}))^{-1}(\psi^{\prime}(x_{0}) - \i k \psi(x_{0}))\,,    
\end{equation}
where $\psi(x) \in M_{n}(\mathbb{C})$ is a matrix solution of (\ref{schroed1}) consisting of linearly independent column solutions as in (\ref{lemma1 eq1}). (See also the proof of \cref{prop2} and \cref{app C}.) Any two choices of matrix solutions $\psi$ differ by right multiplication with a matrix from $GL_{n}(\mathbb{C})$, hence (\ref{refl coeff}) is independent thereof. We now reinstate $s\in S^1$ and state the following result.
\begin{lemma}\label{lemma 4.1}
\begin{enumerate}[label=(\roman*), wide]
\item The reflection matrices for truncation at $x$ satisfy
\begin{align}
R(x, s+T) = R(x, s)\,, \label{sym1 r}\\
R(x+L, s) = R(x, s)\,, \label{sym2 r}
\end{align}

thus the winding numbers $W_{x}$ and $W_{s}$ of det $R(x, s)$ as a function of $s$ and $x$, respectively, are well-defined.

\item The numbers $W_x$, $W_s$ are independent of $x$ and $s$ and given by the expressions
$$
\begin{aligned}
& W_{x} = \int_{0}^{T} d s \operatorname{tr} R(x, s)^{*} \partial_{s} R(x, s)\,, \\
& W_{s} = \int_{0}^{L} d x \operatorname{tr} R(x, s)^{*} \partial_{x} R(x, s)\,.
\end{aligned}
$$
\end{enumerate}
\end{lemma}

\begin{remark}
The \cref{sym1 r,sym2 r} depend on just the corresponding periodicity equation (\ref{periodicity1}), (\ref{periodicity2}) respectively.    
\end{remark} 

\begin{proof}
Part (i) follows from the periodicities of $V$, as just stated, but also because the plane waves (\ref{scattering waves}) take the values $\psi_{R / L}$ at the truncation point $x_{0}$, no matter where it is located. As for part (ii), the independence on $s, x$ is by continuity and the expressions follow from $d(\log \det R) = \tr(R^{*} d R)$.
\end{proof}
The following result complements \cref{prop2}.

\begin{proposition}\label{prop 4.2}
\begin{align}
\ch P_{x} &= W_{x}\,, \qquad \text{(any $x$)}\,, \label{prop 4.2 1}\\
\ch P_{s} &= -W_{s}\,,  \qquad\text{(any $s$})\,. \label{prop 4.2 2}
\end{align}
\begin{proof}
\cref{prop 4.2 1} is a restatement of \cite[Thm. 2]{ortelli} and is based on Lm. 3 there. We recall it briefly. The base space of the bundle $P_{x}$ is the 2-torus $\mathbb{T}_{x}= \gamma \times \mathbb{S}^{1}\times\{x\}\ni (z,s,x)$, where $\gamma$ is the loop described below \cref{3torus}. On that torus, a global section of the frame bundle of $P_{x}$, called $\psi(z,s)=\psi(z, s, x) \in M_{n}(\mathbb{C})$ would exist, were it not for isolated points $(\mu, s_{*})$ along the line $z = \mu$ in $\mathbb{T}_{x}$, where $\psi(\mu, s_*, x_{0}) = 0$, yet $\psi^{\prime}(\mu, s_*, x_{0}) \neq 0$. We set
\begin{equation*}
K(z, s, x)= \psi^{\prime}(\bar{z}, s, x)^{*} \psi^{\prime}(z, s, x)\,,\qquad L(z, s, x)= \psi^{\prime}(\bar{z}, s, x)^{*} \psi(z, s, x)\,.
\end{equation*}
The two matrices are Hermitian for $z$ real; moreover, for $s$ near any of the above points $s_{*}$, $K(\mu, s, x)\equiv K(s, x) >
0$ and $s \mapsto L(\mu, s, x) \equiv L(s, x)$ has an eigenvalue branch $s \mapsto \lambda(s)$ with $\lambda(s_*) = 0$. Also,
$$
\ch P_{x} = \sum_{s_{*}} (\pm 1)\,,
$$
with the upper sign applying in case the crossing is ascending. For short, the Chern number is given by the spectral flow of $s \mapsto L(s, x)$. Finally, the reflection matrix (\ref{refl coeff}) reads
$$
R(s, x) = -(K(s, x) + \i k L(s, x))^{-1}(K(s, x) - \i k L(s, x))\,.
$$
Thus, its eigenvalues $r \in \mathbb{C}$, $(|r| = 1)$ cross $-1$ (counterclockwise) precisely when those of $L(s, x)$, i.e. $\lambda\in\rr$, cross 0 (ascending). That proves (\ref{prop 4.2 1}).\\ 

The argument for (\ref{prop 4.2 2}) is the same upon interchanging $s$ and $x$, except for the opposite orientations of $\mathbb{T}_{x}$ and $\mathbb{T}_{s}$, as seen below (\ref{3.8a}). Hence the minus sign in (\ref{prop 4.2 2}).
\end{proof}
\end{proposition}
Hatted and unhatted symbols shall refer to the reference frames $\hat{F}$ and $F$, cf. \cref{section 2}. Truncation to $\hat{x} \geq\hat{x}_{0}$ of the potential $\hat{V}(\hat{x}, s)$ yields
$$
\theta(\hat{x} - \hat{x}_{0}) \hat{V}(\hat{x}, s) = \theta(x - (\hat{x}_{0} + vs)) V(x, s)
$$
by (\ref{boost}). Thus
\begin{equation}\label{boost r}
\hat{R}(\hat{x}, s) = R(\hat{x} + vs, s)\,,
\end{equation}
after replacing $\hat{x}_{0}$ by $\hat{x}$. As that agrees with the transformation law of $V$, cf. (\ref{transf potential}), \cref{lemma 0} applies to $\hat{R}$ as well.
\begin{theorem}\label{theorem 4.3}
Suppose the two frames are related by (\ref{boost}, \ref{periods}). Then
\begin{align}
& \hat{W}_{x}=m W_x+n W_s\,, \label{galilei edge 1} \\
& \hat{W}_s=m W_s\,. \label{galilei edge 2}
\end{align}
\begin{proof}
As earlier on, we may assume $n=m=1$ in (\ref{periods}), and hence (\ref{base case}).
\cref{boost r} implies
\begin{align}
& \partial_{s} \hat{R}=\partial_{s} R+v \partial_{x} R\,, \label{der edge 1}\\
& \partial_{x} \hat{R}=\partial_{x} R\,, \label{der edge 2}
\end{align}
where the evaluation of the left and right sides is at $(\hat{x}, s)$ and $(\hat{x} + vs,s)$, respectively; likewise, after multiplication with $\hat{R}^*$ and $R^*$. By (ii) of \cref{lemma 4.1}, including the independence clause, we have
\begin{align*}
\hat{W}_x =\frac{1}{L} \int_{0}^{L} d \hat{x} \hat{W} _{\hat{x}}&=\frac{1}{L} \int_{0}^{L} d \hat{x} \int_{0}^{T} d s \tr \hat{R}(\hat{x}, s)^{*} \partial_{s} \hat{R}(\hat{x}, s) \\
& =\int_{0}^{T} d s \frac{1}{L} \int_{0}^{L} d x \tr R^{*}(\partial_{s} R+v \partial_{x} R)\,,
\end{align*}
where we changed the order of integration and performed the change of variable $x := \hat{x} + v s$, which in turn relied on the periodicity (\ref{sym2 r}) by way of $\int_{vs}^{L + vs} d x \ldots = \int_{0}^{L} d x \ldots$. We then split the integral in two and undo the change of order in the first one, so as to obtain
\begin{equation*}
\hat{W}_{x} =\frac{1}{L} \int_{0}^{L} d x W_{x} + \frac{v}{L} \int_{0}^{T} d s W_{s} =W_{x} + \frac{v T}{L} W_{s}\,.
\end{equation*}
That proves (\ref{galilei edge 1}) by (\ref{base case}). As for (\ref{galilei edge 2}), it follows likewise from (\ref{der edge 2}).
\end{proof}
\end{theorem}
The proof of (\ref{galilei edge 1}) ($n=m=1$) is visualized in \cref{fig2}. Each line covers a period and comes with a winding number for $R(x,s)$. The slanted line (winding $\hat{W}_x$) can be deformed to the combined vertical and horizontal lines ($W_x$ and $W_s$). Note the similarity to \cref{fig1}.
\begin{figure}[H]
\centering
\setlength{\unitlength}{4144sp}%
\begin{picture}(1332,1666)(121,-1205)
\thinlines
{\color[rgb]{0,0,0}\multiput(1171,119)(0.00000,-9.00000){121}{\makebox(1.5875,11.1125){\small.}}
}%
\put(1396,299){\makebox(0,0)[lb]{\smash{\fontsize{12}{14.4}\normalfont {\color[rgb]{0,0,0}$\hat{x}=0$}%
}}}
\put(250, 29){\makebox(0,0)[lb]{\smash{\fontsize{12}{14.4}\normalfont {\color[rgb]{0,0,0}$T$}%
}}}
\put(1081,-1141){\makebox(0,0)[lb]{\smash{\fontsize{12}{14.4}\normalfont {\color[rgb]{0,0,0}$L$}%
}}}
{\color[rgb]{0,0,0}\put(450,-960){\vector( 2, 3){900}}
}%
{\color[rgb]{0,0,0}\put(450,-960){\vector( 1, 0){990}}
}%
{\color[rgb]{0,0,0}\put(450,-960){\vector( 0, 1){1350}}
}%
\put(1396,-1141){\makebox(0,0)[lb]{\smash{\fontsize{12}{14.4}\normalfont {\color[rgb]{0,0,0}$x$}%
}}}
\put(270,300){\makebox(0,0)[lb]{\smash{\fontsize{12}{14.4}\normalfont {\color[rgb]{0,0,0}$s$}%
}}}
\put(136,-470){\makebox(0,0)[lb]{\smash{\fontsize{12}{14.4}\normalfont {\color[rgb]{1,0,0}$W_x$}%
}}}
\put(830,-570){\makebox(0,0)[lb]{\smash{\fontsize{12}{14.4}\normalfont {\color[rgb]{1,0,0}$\hat{W}_x$}%
}}}
\put(680,210){\makebox(0,0)[lb]{\smash{\fontsize{12}{14.4}\normalfont {\color[rgb]{1,0,0}$W_s$}%
}}}
\thicklines
{\color[rgb]{1,0,0}\put(451,-961){\vector( 2, 3){720}}
}%
{\color[rgb]{1,0,0}\put(451,-961){\vector( 0, 1){1080}}
}%
{\color[rgb]{1,0,0}\put(451,119){\vector( 1, 0){720}}
}%
\end{picture}%
\caption{Winding numbers for existing and transported
charges.}\label{fig2}
\end{figure}
We conclude with the main result of this section.
\begin{proof}[Proof of \cref{theorem2} (alternative)]
The result follows from \cref{prop 4.2,theorem 4.3}.
\end{proof}
\begin{remark}
\cref{chern3,charge density,prop 4.2 2} combined state that
\begin{equation}\label{density edge}
\int_{\mathbb{R} / L \mathbb{Z}} d x P(x, x) = -W_s
\end{equation}
where $P$ is the Fermi projection. The l.h.s. is $\rk P$. Thus (\ref{density edge}) is a version of Sturm's oscillation theorem, cf. \cite[Thm. XIII.8]{reedsimoniv} and of the Johnson-Moser rotation number \cite{johnson}.
\end{remark}

\appendix
\begin{appendices}
\section{}
Let us at first consider a trivial Hermitian vector bundle $E=M\times V$ with base space $M$ and finite-dimensional fiber $V$. Let $P:M\rightarrow \mathcal{L}(V)$ be a smooth map that takes values in the orthogonal projections on $V$ and let $P$ also denote the subbundle $P\rightarrow M$ of $E$ having fibers $\operatorname{im} P\subset V$. Then $P$ may be non-trivial. For $\operatorname{dim}M =2$ its Chern number is
\begin{equation}\label{appA 1}
\ch P=\frac{\i}{2\pi} \int_M \text{tr}(P(dP\wedge dP))\,.
\end{equation}
The purpose of this appendix is (a) to extend the result to bundles that are not Hermitian and (b) to express their Chern numbers in terms of frames that are not orthogonal. Actually, (a) is evident, since the notion of trace rests on duality and not on an inner product, which can though be used to compute \cref{appA 1} in terms of orthogonal frames.\\

We thus recall dual systems $(\widetilde{V},V,b)$, where $\widetilde{V}$ and $V$ are finite-dimensional vector spaces and $b:\widetilde{V}\times V\rightarrow\cc$ is a non-degenerate bilinear form, meaning
\begin{equation*}
\begin{aligned}
& b(\widetilde{v}, v)=0\,,\enskip(v \in V) &\Longrightarrow\enskip \widetilde{v}=0\,, \\
& b(\widetilde{v}, v)=0\,,\enskip(\widetilde{v} \in \widetilde{V}) &\Longrightarrow \enskip v=0 \,.
\end{aligned}
\end{equation*}
Thus, $b$ induces an isomorphism
\begin{equation}\label{appA 2}
V^*\cong\widetilde{V}\,,    
\end{equation}
where $V^*$ is the dual space of $V$. In particular, $\operatorname{dim}\widetilde{V}=\operatorname{dim}V$. Moreover, any subspace $\widetilde{V}'\subset\widetilde{V}$ defines a subspace $\widetilde{V}'^{\perp}\subset V$ by
\begin{equation}\label{appA 3}
\widetilde{V}'^{\perp}=\{v\in V \mid b(\widetilde{v}', v)=0,\,(\widetilde{v}' \in \widetilde{V}')\} \,.
\end{equation}
It has $\operatorname{dim}\widetilde{V}'^{\perp}=\operatorname{codim}\widetilde{V}'$. Let $(\widetilde{V}',V',b)$ be a dual subsystem of $(\widetilde{V},V,b)$ , i.e. $\widetilde{V}'\subset \widetilde{V}$ and $V'\subset V$ be subspaces such that the restriction of $b$ to $\widetilde{V}'\times V'$ is non-degenerate. The dual subspace is uniquely characterized by a projection $P:V\rightarrow V$ such that
\begin{equation*}
\operatorname{im}P=V'\,, \qquad \operatorname{ker}P=\widetilde{V}'^{\perp}\,. 
\end{equation*}
The claim is a consequence of (i) $V'\cap\widetilde{V}'^{\perp}=\{0\}$, (ii) $V'+\widetilde{V}'^{\perp}=V$, which are in turn seen as follows: (i) Any $v\in V'\cap \widetilde{V}'^{\perp}$ is $v=0$ by (\ref{appA 3}) and the non-degeneracy of the restriction of $b$; (ii) then follows on grounds of dimension. As a map $P:V\rightarrow V$, the projection may be seen as a map $P\in V\otimes V^*\cong V\otimes\widetilde{V}$, where the isomorphism is by (\ref{appA 2}). In those terms, $P$ can be constructed by means of a (non-unique) biorthogonal frame, i.e. bases
\begin{equation}\label{appA 4}
\widetilde{F}=(\widetilde{v}_1,\dotsc,\widetilde{v}_n),\qquad F=(v_1,\dotsc,v_n)    
\end{equation}
of $\widetilde{V}'$ and $V'$, respectively, such that
\begin{equation}\label{appA 5}
b(\widetilde{v}_i,v_j)=\delta_{ij}\,.    
\end{equation}
Then $P$ is given by
\begin{equation*}
P=\sum_{i=1}^nv_i\otimes \widetilde{v}_i\,.    
\end{equation*}
Let next $\widetilde{E}=M\times \widetilde{V}$ and $E=M\times V$ be two trivial bundles and $(\widetilde{V},V,b)$ be a dual system. Let $\widetilde{P}\rightarrow M$ and $P\rightarrow M$ be (possibly non-trivial) subbundles of $\widetilde{E}$ and $E$, respectively, such that $b$ remains fiberwise non-degenerate.\\

Let bases $\widetilde{F}=\widetilde{F}(\widetilde{p})$ and 
$F=F(p)$ of fibers $\widetilde{P}_{\widetilde{p}}$ and $P_p$ be given as in (\ref{appA 4}, \ref{appA 5}), locally in $\widetilde{p},p\in M$. Let then $b(\widetilde{F},F)$ be a matrix of order $n$, defined (likewise locally) by
\begin{equation*}
b(\widetilde{F}(\widetilde{p}),F(p))_{ij}=b(\widetilde{v}_i(\widetilde{p}),v_j(p))\,, \qquad (i,j=1,\dotsc,n;\enskip\widetilde{p},p\in M)\,.    
\end{equation*}
We then have:
\begin{proposition}\label{prop appA}
Let $\operatorname{dim}M=2$. The Chern number of $P$ is given by \cref{appA 1} as well as by
\begin{equation}\label{appA 7}
\ch P=\frac{\i}{2 \pi} \int_M \tr\bigl(\widetilde{d}\wedge d\bigr) b(\widetilde{F}, F) \Big|_{\widetilde{p}=p}\,.
\end{equation}
\begin{proof}
The integrand is independent of bi-orthogonal frames on the intersection of patches. Hence it is globally defined.\\

It pays to spell out the integrand. We consider functions $f:M\times M\rightarrow \cc$ and their restriction to the diagonal submanifold $\{(\widetilde{p},p)\mid \widetilde{p}=p\}$. Also, given a vector field $X$ on $M$, vector fields $\widetilde{X}=(X,0)$ and $X\equiv(0,X)$ are induced on $M\times M$. Then 
\begin{equation}\label{a8}
(\widetilde{d} \wedge d) f\big|_{\widetilde{p}=p}(X, Y)\coloneqq g(p)\,,\qquad g=(\widetilde{X} Y-\widetilde{Y} X) f\big|_M\,.
\end{equation}
Given a section $v=v(p)\in V$ and $X$ as above, a section $d_X v$ is defined as a directional derivative, i.e.
\begin{equation*}
(d_X v)(p)=\frac{d}{d \lambda} v(p(\lambda))\Big|_{\lambda=0}\,,
\end{equation*}
where $\lambda\mapsto p(\lambda)$ is a curve in $M$ with $p(0)=p$, $dp/d\lambda\mid_{\lambda=0}=X_p$. (The construction is well-defined because the vector space $V$ is independent of $p$.)\\

A first claim is now 
\begin{equation}\label{a9}
((\widetilde{d} \wedge d) \operatorname{tr} b(\widetilde{F}, F))\big|_{\widetilde{p}=p}(X, Y)=\sum_{i=1}^n b((d_X \widetilde{v}_i)(p),(d_Y v_i)(p))-(X\leftrightarrow Y)\,.
\end{equation}
It is seen from (\ref{a8}) and 
\begin{equation*}
\begin{gathered}
\operatorname{tr} b(\widetilde{F}, F)=\sum_{i=1}^n b(\widetilde{v}_i(\widetilde{p}), v_i(p))\,,\\
(\widetilde{X}Y\operatorname{tr} b(\widetilde{F}, F))(\widetilde{p}, p)= \sum_{i=1}^n b(d_X \widetilde{v_i}(\widetilde{p}), d_Y v_i(p))\,.
\end{gathered}
\end{equation*}
A second claim, this time referring to (\ref{appA 1}), is
\begin{equation}\label{a10}
\operatorname{tr}(P(d P \wedge d P))(X, Y)=\sum_{i=1}^n\langle d_X \widetilde{v}_i, d_Y v_i\rangle-(X \leftrightarrow Y)\,,
\end{equation}
where $\widetilde{v}_i$ is the dual basis to $v_i$ and $\langle\cdot,\cdot\rangle$ the duality bracket. It is seen from 
\begin{equation*}
\begin{gathered}    
P=\sum_{i=1}^n v_i\langle \widetilde{v}_i, \cdot\rangle\,,\\
d_X P=\sum_{i=1}^n d_X v_i\langle\widetilde{v}_i, \cdot\rangle+v_i\langle d_X \widetilde{v}_i, \cdot\rangle\,,\\
\operatorname{tr}(P d_X P d_Y P)=\sum_{i=1}^n\langle\widetilde{v}_j, d_X v_i\rangle\langle \widetilde{v}_i, d_Y v_j\rangle+\langle \widetilde{v}_i, d_X v_j\rangle\langle d_Y \widetilde{v}_j, v_i\rangle\\
\hspace{3cm} +\delta_{ij}\langle d_X \widetilde{v}_i, d_Y v_j\rangle+\langle d_X \widetilde{v}_i, v_j\rangle\langle d_Y \widetilde{v}_j, v_i\rangle\,,
\end{gathered}
\end{equation*}
where we made use of $\langle \widetilde{v}_j,v_i\rangle=\delta_{ij}$. The first sum is symmetric in $X,Y$, while the second and fourth terms add up to $(d_X \langle\widetilde{v}_i,v_j\rangle)\langle d_Y\widetilde{v}_j,v_i\rangle=0$. That establishes (\ref{a10}).\\

Finally, $\widetilde{V}$ and $b$ may be taken to be $V^*$ and $\langle\cdot,\cdot\rangle$ to begin with, because of (\ref{appA 2}). Then the hypothesis (\ref{appA 5}) says that $(\widetilde{v}_i)$ is the dual basis. At that point (\ref{a9}) and (\ref{a10}) are seen to agree.
\end{proof}
\end{proposition}

\section{}\label{app b}
\begin{proof}[Proof of \cref{lemma1}]
In this proof we set $V^+_{(z,s,x_0)}\eqqcolon V^+_{x_0}$, since $z,s$ remain fixed. Likewise for $\widetilde{V}^-_{x_0}$.
\begin{enumerate}[label=(\roman*), wide]
\item Given any solution $\psi=\psi(x)$ as in (\ref{schroed1}), so is $\psi_L=\psi_L(x)\coloneqq\psi(x-L)$, because of $V(x-L)=V(x)$. If $\psi$ is $L^2$ at $x=+\infty$, then so is $\psi_L$. The initial data, cf. (\ref{initial cond}), are the same, in the sense that $(\Psi, x_0)=(\Psi_L,x_0+L)$. Thus $V_{x_0}^+\subset V^+_{x_0+L}$, and viceversa by $L\rightarrow -L$.
\end{enumerate}
In the remainder of this proof, we shall also drop the subscript $x_0$.
\begin{enumerate}[resume,label=(\roman*), wide]
\item Let $V^-\subset \cc^n \oplus\cc^n$ be defined as $V^+$, except that the $L^2$-condition is imposed at $x=-\infty$. We shall prove a more general statement than the first part of (ii), namely
\begin{align}
\dim V^{\pm}&=n \label{b1}\,,\\
V^+\cap V^-&=\{0\} \label{b2}
\end{align}
and hence
\begin{equation}\label{b3}
V^+\oplus V^-=\cc^n\oplus\cc^n\,.    
\end{equation}
The second part of (ii) follows similarly.\\

The Hamiltonian $H$, cf. (\ref{hamiltonian pos}), is essentially self-adjoint on its minimal domain $C_0^{\infty}(\rr,\cc^n)$, hence it has deficiency indices $\langle0,0\rangle$. Let $H_+$ be the same operator $H$ but now acting on $L^2(\rr_+,\cc^n)$ with corresponding minimal domain; or rather its closure. It is a symmetric operator with deficiency indices $\langle n,n\rangle$, i.e.
\begin{equation}\label{b4}
\dim\ker (z-H_+^*)=n    
\end{equation}
for $\Im z\gtrless 0$. Actually, this holds true for any $z\notin \sigma(H)$ and rests on the same stability argument that underlies the constancy of (\ref{b4}) on the upper and lower half-plane (see e.g. the proof of \cite[Thm. X.1]{reedsimonii}) or even in a slit plane (Corollary thereof). It rests on Eq. (X.1) there, viz.
\begin{equation*}
\Vert(z-H_+)\varphi)\Vert\geq\rho\Vert\varphi\Vert\,,
\end{equation*}
with $\rho>0$ and for all $\varphi\in D(H_+)$, which holds true here for $z\notin\sigma(H)$ because $D(H_+)\subset D(H)$ under the natural embedding $L^2(\rr_+,\cc^n)\subset L^2(\rr,\cc^n)$. That proves (\ref{b4}) for $z$ as stated. Since $H_+^*$ is given by the same formal expression as $H$, and in fact on $\rr_+$ yet without boundary conditions, \cref{b4} is equivalent to the $+$ case of (\ref{b1}).\\

As for (\ref{b2}), any $\Psi$ in that intersection is the initial data of a solution $\psi$ of $H(s)\psi=z\psi$ as an equation in $L^2$, thus $\psi=0$ by $z\notin H(s)$.
\item We have to show the following implication to be valid for $\Psi\in\widetilde{V}^-$:
\begin{equation}\label{b5}
\omega(\widetilde{\Psi},\Psi)=0\,, \enskip (\Psi\in V^+_{(z,s,x_0)})\implies\widetilde{\Psi}=0 \,;    
\end{equation}
as well as a similar one valid for $\Psi\in V^+$, having an analogous proof.\\

We claim that the hypothesis of (\ref{b5}) implies itself and more generally for $\Psi\in\cc^n\oplus\cc^n$ instead of $\Psi\in V^+$. Once this is established, \cref{b5} follows from the non-degeneracy of $\omega$ on $\cc^n\oplus\cc^n$. To show our claim, we use (\ref{b3}). Any $\Psi\in\cc^n\oplus\cc^n$ is thus decomposable as $\Psi=\Psi^++\Psi^-$, with $\Psi^{\pm}\in V^{\pm}$, and we have
\begin{equation*}
\omega(\widetilde{\Psi},\Psi)=\omega(\widetilde{\Psi},\Psi^+)+\omega(\widetilde{\Psi},\Psi^-)=0
\end{equation*}
\end{enumerate}
by the original hypothesis of (\ref{b5}) and by evaluating the Wronskian (\ref{3.7'}) at $x=-\infty$.
\end{proof}

\begin{proof}[Proof of \cref{prop 3.5}]
Let $C:\cc\rightarrow\cc$, $z\mapsto\bar{z}$ be complex conjugation and $\overline{\mathbb{T}}_3=\gamma\times S^1\times S^1$, cf.~(\ref{3torus}). (One could pick $\gamma$ so that $\gamma=\overline{\gamma}$, hence $\overline{\mathbb{T}}_3$=$\mathbb{T}_3$, but this is not necessary.) As in the proof of \cref{lemma1}, let
\begin{equation}\label{potentials app b}
V^{\pm}\equiv V^{\pm}_{(z,s,x_0)}\,,\qquad \widetilde{V}^{\pm}\equiv \widetilde{V}^{\pm}_{(z,s,x_0)}
\end{equation}
be the spaces of solutions $\psi=\psi(x)$ and $\widetilde{\psi}=\widetilde{\psi}(x)$ of (\ref{schroed1}, \ref{schroed2}), respectively, that decay at $x\rightarrow\pm\infty$, cf.~(\ref{lemma1 eq1}). As before, $\psi(x), \widetilde{\psi}(x)\in\cc^n$ consist of column and row vectors. The corresponding vector bundles on $\mathbb{T}_3$ are denoted $P_V^{\pm}\subset E$, $\widetilde{P}_V^{\pm}\subset\widetilde{E}$, cf. (\ref{3.10a}), where the subscript emphasizes that they depend on the potentials $V=V(x,s)$ because their fibers (\ref{potentials app b}) do.\\

Transposition $T$ and complex conjugation $K$ act on $\cc^n$, with the former turning column into row vectors. They act pointwise on $\psi$, $\widetilde{\psi}$, as well as on matrix-valued potentials $V$. By $V=V^*$ we have $V^T=\overline{V}$.\\

We claim that $T$, $K$ induce bundle (anti-)isomorphisms
\begin{equation}\label{commutative diag}
\begin{tikzcd}
    P_V^\pm & \widetilde{P}_{V^T}^\pm \\
	\mathbb{T}_3 & \mathbb{T}_3
	\arrow["T",from=1-1, to=1-2]
	\arrow[from=1-1, to=2-1]
	\arrow[from=1-2, to=2-2]
	\arrow["\text{id}",from=2-1, to=2-2]
\end{tikzcd}\,,
\qquad\qquad
\begin{tikzcd}
    P_V^\pm & P_{\overline{V}}^\pm \\
	\mathbb{T}_3 & \overline{\mathbb{T}}_3
	\arrow["K",from=1-1, to=1-2]
	\arrow[from=1-1, to=2-1]
	\arrow[from=1-2, to=2-2]
	\arrow["C",from=2-1, to=2-2]
\end{tikzcd}\,,
\end{equation}		
where the vertical arrows are the bundle projections. This follows by taking the transpose and the complex conjugation of (\ref{schroed1}), which read
\begin{equation*}
-(\psi^T)''(x)+\psi^T(x)V^T(x,s)=z\psi^T(x)\,,\qquad -\overline{\psi}''(x)+\overline{V}(x,s)\overline{\psi}(x)=\bar{z}\overline{\psi}(x)\,.
\end{equation*}
We next consider restrictions of the above bundles to a fixed $z$, as done already for $P_V^+$ in (\ref{third chern number}). First, we compare different bundles related to the same potential $V$ by claiming 
\begin{equation*}
\ch(P_V^+)_z=-\ch(P_V^-)_z\,,\qquad \ch(P_V^+)_z=-\ch(\widetilde{P}_V^-)_z\,.
\end{equation*}
The first equation arises because $P_V^+\oplus P_V^-=\mathbb{T}_3\times (\cc^n\oplus\cc^n)$ is trivial, cf. (\ref{b3}); the second, because the two sides are related by an interchange of $d$ and $\widetilde{d}$ in (\ref{appA 7}). We conclude that the r.h.s. of the two equations are equal. Likewise,
\begin{equation}\label{b8}
\ch(P_V^+)_z=\ch(\widetilde{P}_V^+)_z\,.
\end{equation}
Second, we compare different potentials $V$ by claiming
\begin{equation*}
\ch(P_V^+)_z=\ch(\widetilde{P}_{V^T}^+)_z\,,\qquad \ch(P_V^+)_z=-\ch(P_{\overline{V}}^+)_{\bar{z}}\,.
\end{equation*}
This follows from  (\ref{commutative diag}) and the anti-linearity of the second morphism. The minus sign comes from the reality of Chern numbers, i.e. $\ch(P_V^+)=\overline{\ch(P_V^+)}$, from $\bar{i}=-i$ in (\ref{appA 7}) and from $C$ not affecting the orientation $ds\wedge dx$. Finally, using also (\ref{b8}), we find
\begin{equation*}
\ch(P_V^+)_z=\ch(\widetilde{P}_V^+)_z=\ch(P_{V^T}^+)_z=\ch(P_{\overline{V}}^+)_z=-\ch(P_V^+)_{\bar{z}}\,.
\end{equation*}
That yields $\ch(P_V^+)_z=0$ first for $z=\bar{z}$, such as for $z=\mu$, and then for all $z\in\gamma$, since it is constant. 
\end{proof}

\section{}\label{app C}
We shall derive \cref{green} by providing two lemmas of independent interest. They expand the derivation of \cite[Eq.$\,$(35)]{ortelli}. To this end we rename the matrix solutions $\psi(x),\widetilde{\psi}(x)\in M_n(\cc)$ of (\ref{schroed1}) and (\ref{wronsk}) as $\psi_+(x)$, $\widetilde{\psi}_-(x)$ respectively, so as to remind the reader of the end of the line, $x\rightarrow\pm\infty$, at which they decay (\ref{lemma1 eq1}). Consistently, (\ref{3.7'}, \ref{gauge fix}) now reads 
\begin{equation*}
W(\widetilde{\psi}_-,\psi_+)=\mathds{1}\,.    
\end{equation*}
We complement those matrix solutions by two more, $\psi_-(x)$ and $\widetilde{\psi}_+(x)$, which differ from the above in that they decay at the opposite ends. Let their normalization be $W(\widetilde{\psi}_+,\psi_-)=-\mathds{1}$, cf. \cite[Eq.$\,$(36)]{ortelli}. Clearly, $W(\widetilde{\psi}_{\pm},\psi_{\pm})=0$ because the (constant) Wronskian can be evaluated at $x\rightarrow\pm\infty$, where both solutions decay. These Wronskians can be summarized in terms of the matrices
\begin{equation}\label{c1}
\underline{\psi}=\begin{pmatrix}
\psi_{+}, \psi_{-}\end{pmatrix}\,, \qquad \underline{\widetilde{\psi}}=\begin{pmatrix}
\widetilde{\psi}_{+} \\
\widetilde{\psi}_{-}
\end{pmatrix}
\end{equation}
of sizes $n\times 2n$ and $2n\times n$, and in fact by 
\begin{equation}\label{c2}
\underline{\widetilde{\psi}}\, \underline{\psi}^{\prime}-\widetilde{\underline{\psi}}'\, \underline{\psi}=-\omega\,,
\end{equation}
where $\omega$ is the matrix of order $2n$ underlying the symplectic form (\ref{bilinear form}). Indeed, the l.h.s. of (\ref{c2}) equals
\begin{equation*}
W(\underline{\widetilde{\psi}}, \underline{\psi})=\begin{pmatrix}
W(\widetilde{\psi}_{+},\psi_{+}) & W(\widetilde{\psi}_{+},\psi_{-}) \\
W(\widetilde{\psi}_{-},\psi_{+}) & W(\widetilde{\psi}_{-}, \psi_{-})\end{pmatrix}=-\begin{pmatrix}
0 & \mathds{1}_n \\
-\mathds{1}_n & 0
\end{pmatrix}\,.
\end{equation*}
\begin{lemma}
Any solution $\psi=\psi(x)\in\cc^n$ of the ODE (\ref{schroed1}) is determined by its initial values $\psi(y),\psi^{\prime}(y)$ at $x=y$ as follows:
\begin{equation}\label{c3}
\begin{pmatrix}
\psi(x) \\
\psi^{\prime}(x)
\end{pmatrix}=P(x, y)\begin{pmatrix}
\psi(y) \\
\psi^{\prime}(y)
\end{pmatrix}\,,
\end{equation}
where the (spatial) propagator is given by
\begin{equation}\label{c4}
P(x,y)=\begin{pmatrix}
-\underline{\psi}(x) \omega \widetilde{\underline{\psi}}'(y) & \underline{\psi}(x) \omega \underline{\widetilde{\psi}}(y) \\
-\underline{\psi}'(x) \omega \widetilde{\underline{\psi}}'(y) & \underline{\psi}'(x) \omega \underline{\widetilde{\psi}}(y)
\end{pmatrix}\,.
\end{equation}
\begin{proof}
Let $P$ be defined by (\ref{c4}). By linearity it suffices to prove (\ref{c3}) for $\psi$  being any column vector of $\psi_+$ or $\psi_-$, which is tantamount to replacing $\psi$ by $\underline{\psi}$ in the claim, cf. (\ref{c1}), and thus $\psi^{\prime}$ by $\underline{\psi}^{\prime}$ as well. The first block row of the r.h.s. then equals     
\begin{equation*}
-\underline{\psi}(x) \omega \widetilde{\underline{{\psi}}}'(y) \underline{\psi}(y)+\underline{\psi}(x) \omega \underline{\widetilde{\psi}}(y) \underline{\psi}^{\prime}(y) =-\underline{\psi}(x) \omega(\widetilde{\underline{{\psi}}}'(y) \underline{\psi}(y)-\underline{\widetilde{\psi}}(y) \underline{\psi}^{\prime}(y)) = -\underline{\psi}(x)\omega^2 =\underline{\psi}(x)
\end{equation*}
by (\ref{c2}). The second row just differs in that $\underline{\psi}(x)$ is replaced by $\underline{\psi}^{\prime}(x)$.
\end{proof}
\end{lemma}
\begin{lemma}
The Green's function $G(x,x')\equiv G(x,x';z)$ of the Hamiltonian (\ref{hamiltonian pos}) at $z\notin \sigma(H)$ is \begin{equation}\label{c5}
G(x, x^{\prime})=-\theta(x-x^{\prime}) \psi_{+}(x) \widetilde{\psi}_{-}(x^{\prime})-\theta(x^{\prime}-x) \psi_{-}(x) \widetilde{\psi}_{+}(x^{\prime})\,.
\end{equation}   
\end{lemma}
\begin{remark}
\cref{green} is a special case thereof.   
\end{remark}
\begin{proof}
Through the definition of the resolvent, i.e. $(H-z)R(z)=\mathds{1}$, that of the Green's function amounts to
\begin{gather}
\Big(-\frac{\partial^2}{\partial x^2}+V(x)-z\Big) G(x, x^{\prime})=0\,, \qquad (x\neq x')\,,\label{c6} \\
G(x, x^{\prime})\big|_{x=x^{\prime}-0} ^{x=x'+0}=0\,,\label{c7} \\
-\frac{\partial}{\partial x} G(x, x^{\prime})\big|_{x=x^{\prime}-0} ^{x=x^{\prime}+0}=\mathds{1}\,.\label{c8}
\end{gather}
We let $G$ be satisfied by (\ref{c5}) and we shall show that it satisfies all three equations. \cref{c6} is satisfied by the definitions of $\psi_{\pm}(x)$. To verify the remaining ones, we first note that $P(x,x)=\mathds{1}$ implies, by its second block column, cf. (\ref{c4}),
\begin{align*}
\psi_{+}(x) \widetilde{\psi}_{-}(x)-\psi_{-}(x) \widetilde{\psi}_{+}(x)&=0\,, \\
\psi_{+}^{\prime}(x) \widetilde{\psi}_{-}(x)-\psi_{-}^{\prime}(x) \widetilde{\psi}_{+}(x)&=\mathds{1}\,.    
\end{align*}
\cref{c7,c8} are now satisfied because of (\ref{c5}) and in view of
\begin{equation*}
-\frac{\partial}{\partial x}G(x,x^{\prime})=\begin{cases}
\psi_{+}^{\prime}(x) \widetilde{\psi}_{-}(x)\,,\qquad (x=x^{\prime}+0)\,, \\
\psi_{-}^{\prime}(x) \widetilde{\psi}_{+}(x)\,, \qquad(x=x^{\prime}-0)\,.
\end{cases}
\end{equation*}
\end{proof}
\end{appendices}
\textbf{Acknowledgements.} We thank A. Bols, M. Delladio and R. Seiler for useful discussions. Support by the Swiss National Science Foundation under project funding ID: TMAG-2\textunderscore 209376 and ID: 200021\textunderscore 207537 is acknowledged.\\

\textbf{Data availability and conflict of interest.} Data availability is not applicable to this article as no data were created or analyzed in this study. There are no known conflicts of interest.

\nocite{*}
\bibliography{biblio}
\end{document}

%% file: 3torusUpdated.pdf_t
\begin{picture}(0,0)%
\includegraphics{3torusUpdated}%
\end{picture}%
\setlength{\unitlength}{4144sp}%
\begingroup\makeatletter\ifx\SetFigFont\undefined%
\gdef\SetFigFont#1#2#3#4#5{%
  \reset@font\fontsize{#1}{#2pt}%
  \fontfamily{#3}\fontseries{#4}\fontshape{#5}%
  \selectfont}%
\fi\endgroup%
\begin{picture}(2005,1527)(-339,-1025)
\put(361,-961){\makebox(0,0)[lb]{\smash{{\SetFigFont{12}{14.4}{\familydefault}{\mddefault}{\updefault}{\color[rgb]{0,.56,0}$\mathbb{T}_z$}%
}}}}
\put(1556,-891){\makebox(0,0)[lb]{\smash{{\SetFigFont{12}{14.4}{\familydefault}{\mddefault}{\updefault}{\color[rgb]{0,0,0}$z$}%
}}}}
\put(1651,-311){\makebox(0,0)[lb]{\smash{{\SetFigFont{12}{14.4}{\familydefault}{\mddefault}{\updefault}{\color[rgb]{0,0,0}$s$}%
}}}}
\put(-224,379){\makebox(0,0)[lb]{\smash{{\SetFigFont{12}{14.4}{\familydefault}{\mddefault}{\updefault}{\color[rgb]{0,0,0}$x$}%
}}}}
\put(1261,254){\makebox(0,0)[lb]{\smash{{\SetFigFont{12}{14.4}{\familydefault}{\mddefault}{\updefault}{\color[rgb]{0,0,0}$\mathbb{T}_3$}%
}}}}
\put(1351,-106){\makebox(0,0)[lb]{\smash{{\SetFigFont{12}{14.4}{\familydefault}{\mddefault}{\updefault}{\color[rgb]{0,0,1}$\mathbb{T}_x$}%
}}}}
\put(136,209){\makebox(0,0)[lb]{\smash{{\SetFigFont{12}{14.4}{\familydefault}{\mddefault}{\updefault}{\color[rgb]{1,0,0}$\mathbb{T}_s$}%
}}}}
\end{picture}%

%% file: main.bbl
\providecommand{\noopsort}[1]{}\providecommand{\singleletter}[1]{#1}%
\begin{thebibliography}{10}

\bibitem{az}
A.~Altland and M.~R. Zirnbauer.
\newblock Nonstandard symmetry classes in mesoscopic normal-superconducting
  hybrid structures.
\newblock {\em Phys. Rev. B}, 55:1142--1161, 1997.

\bibitem{ortelli}
G.~Br{\"a}unlich, G.~M. Graf, and G.~Ortelli.
\newblock Equivalence of topological and scattering approaches to quantum
  pumping.
\newblock {\em Commun. Math. Phys.}, 295(1):243--259, 2010.

\bibitem{buettiker}
M.~B{\"u}ttiker, H.~Thomas, and A.~Pr{\^e}tre.
\newblock Current partition in multiprobe conductors in the presence of slowly
  oscillating external potentials.
\newblock {\em Zeitschrift f{\"u}r Physik B Condensed Matter}, 94(1):133--137,
  1994.

\bibitem{Citro}
R.~Citro and M.~Aidelsburger.
\newblock Thouless pumping and topology.
\newblock {\em Nature Reviews Physics}, 5(2):87--101, 2023.

\bibitem{AvronDanaZaK}
I~Dana, Y~Avron, and J~Zak.
\newblock Quantised hall conductance in a perfect crystal.
\newblock {\em Journal of Physics C: Solid State Physics}, 18(22):L679, 1985.

\bibitem{gawedzki}
M.~Fruchart, D.~Carpentier, and K.~Gawedzki.
\newblock Parallel transport and band theory in crystals.
\newblock {\em Europhysics Letters}, 106(6):60002, 2014.

\bibitem{hasan}
M.~Z. Hasan and C.~L. Kane.
\newblock Colloquium: Topological insulators.
\newblock {\em Rev. Mod. Phys.}, 82:3045--3067, 2010.

\bibitem{johnson}
R.~Johnson and J.~Moser.
\newblock {The rotation number for almost periodic potentials}.
\newblock {\em Commun. Math. Phys.}, 84(3):403 -- 438, 1982.

\bibitem{kitaev}
A.~Kitaev.
\newblock Periodic table for topological insulators and superconductors.
\newblock {\em AIP Conference Proceedings}, 1134, 2009.

\bibitem{cmp/1104180308}
M.~Klein and R.~Seiler.
\newblock {Power-law corrections to the Kubo formula vanish in quantum Hall
  systems}.
\newblock {\em Commun. Math. Phys.}, 128(1):141 -- 160, 1990.

\bibitem{kobayashi}
S.~Kobayashi.
\newblock {\em Differential Geometry of Complex Vector Bundles}.
\newblock Princeton University Press, 1987.

\bibitem{Lohse}
M.~Lohse, C.~Schweizer, O.~Zilberberg, M.~Aidelsburger, and I.~Bloch.
\newblock A {Thouless} quantum pump with ultracold bosonic atoms in an optical
  superlattice.
\newblock {\em Nature Physics}, 12(4):350--354, 2016.

\bibitem{nakahara}
M.~Nakahara.
\newblock {\em Geometry, topology and physics}.
\newblock Graduate student series in physics. Hilger, Bristol, 1990.

\bibitem{Nakajima2021}
S.~Nakajima, N.~Takei, K.~Sakuma, Y.~Kuno, P.~Marra, and Y.~Takahashi.
\newblock Competition and interplay between topology and quasi-periodic
  disorder in {Thouless} pumping of ultracold atoms.
\newblock {\em Nature Physics}, 17(7):844--849, 2021.

\bibitem{Nakajima}
S.~Nakajima, T.~Tomita, S.~Taie, T.~Ichinose, H.~Ozawa, L.~Wang, M.~Troyer, and
  Y.~Takahashi.
\newblock Topological {Thouless} pumping of ultracold fermions.
\newblock {\em Nature Physics}, 12(4):296--300, 2016.

\bibitem{reedsimonii}
M.~Reed and B.~Simon.
\newblock {\em II: Fourier Analysis, Self-Adjointness}.
\newblock Methods of Modern Mathematical Physics. Elsevier Science, 1975.

\bibitem{reedsimoniv}
M.~Reed and B.~Simon.
\newblock {\em IV: Analysis of Operators}.
\newblock Methods of Modern Mathematical Physics. Elsevier Science, 1978.

\bibitem{topins}
A.~P. Schnyder, S.~Ryu, A.~Furusaki, and A.~W.~W. Ludwig.
\newblock Classification of topological insulators and superconductors in three
  spatial dimensions.
\newblock {\em Phys. Rev. B}, 78:195125, 2008.

\bibitem{thouless}
D.~J. Thouless.
\newblock Quantization of particle transport.
\newblock {\em Phys. Rev. B}, 27:6083--6087, 1983.

\bibitem{Walter}
A.-S. Walter, Z.~Zhu, M.~Gächter, J.~Minguzzi, S.~Roschinski, K.~Sandholzer,
  K.~Viebahn, and T.~Esslinger.
\newblock Quantization and its breakdown in a {Hubbard}--{Thouless} pump.
\newblock {\em Nature Physics}, 19(10):1471--1475, 2023.

\bibitem{zhu}
Z.~Zhu, M.~Gächter, A.-S. Walter, K.~Viebahn, and T.~Esslinger.
\newblock Reversal of quantized {Hall} drifts at noninteracting and interacting
  topological boundaries.
\newblock {\em Science}, 384(6693):317--320, 2024.

\end{thebibliography}
